\documentclass[preprint,12pt]{elsarticle}

\usepackage{amsmath, amssymb, amsthm, mathtools, bm}
\usepackage{graphicx}
\usepackage{caption}
\usepackage{hyperref}
\usepackage{lmodern}
\usepackage[margin=1in]{geometry}
\usepackage{physics}
\usepackage{siunitx}
\usepackage{subcaption}
\usepackage{booktabs}
\usepackage{array}
\usepackage{microtype}
\usepackage{enumitem}
\usepackage[table]{xcolor}
\newcommand{\tblhead}[1]{{\footnotesize\textbf{#1}}}
\usepackage{cleveref}
\usepackage{tikz}
\usetikzlibrary{arrows.meta,patterns,decorations.markings,calc}
\usepackage{amssymb}
\usepackage{multirow}
\usepackage{tabularx}            
\newcolumntype{M}{>{$}l<{$}}     
\newcolumntype{Y}{>{\raggedright\arraybackslash}X}  
\usepackage{bbm}
\usepackage{subcaption} 
\renewcommand{\arraystretch}{1.2}                   

\numberwithin{equation}{section}

\theoremstyle{plain}
\newtheorem{theorem}{Theorem}[section]
\newtheorem{lemma}[theorem]{Lemma}
\newtheorem{proposition}[theorem]{Proposition}

\theoremstyle{definition}

\newtheorem{assumption}[theorem]{Assumption}
\usepackage{algorithm}
\usepackage[noend]{algpseudocode}
\usepackage{multicol}

\usepackage{float}              
\usepackage[section]{placeins}  

\newcommand{\R}{\mathbb{R}}
\newcommand{\C}{\mathbb{C}}

\newcommand{\E}{\mathbb{E}}

\newcommand{\HH}{^{\mathsf{H}}} 
\newcommand{\Had}{\odot}        

\newcommand{\rfft}{\mathrm{rfft2}}
\newcommand{\irfft}{\mathrm{irfft2}}
\newcommand{\FFT}{\mathcal{F}}
\newcommand{\IFFT}{\mathcal{F}^{-1}}
\newcommand{\LL}{\mathrm{LL}}
\newcommand{\LH}{\mathrm{LH}}
\newcommand{\HL}{\mathrm{HL}}
\newcommand{\HHH}{\mathrm{HH}}  

\DeclareMathOperator{\dist}{dist}

\journal{Engineering Applications of Artificial Intelligence}

\begin{document}

\begin{frontmatter}



\title{A Physics-Informed Fourier-Wavelet Transformer for Multiscale Computational Fluid Dynamics Surrogate Modeling}


\author[1]{Somyajit Chakraborty}
\ead{chksomyajit@sjtu.edu.cn}

\author[1]{Ming Pan}
\ead{panming@sjtu.edu.cn}

\author[1]{Xizhong Chen\corref{cor1}}
\ead{chenxizh@sjtu.edu.cn}
\cortext[cor1]{Corresponding author}

\affiliation[1]{organization={Shanghai Jiao Tong University},
            department={Department of Chemistry and Chemical Engineering},
            addressline={800 Dongchuan Road, Minhang District},
            city={Shanghai},
            postcode={200240},
            country={China}}

\begin{abstract}
Physics-informed surrogate models can accelerate computational fluid dynamics simulations. However, many existing methods reproduce global flow patterns more reliably than localized multiscale structures. This study presents a physics-informed Fourier-wavelet transformer for next-step velocity-field reconstruction in real-world flow benchmarks. The proposed formulation combines hybrid Fourier-wavelet spectral encoding with physics-biased self-attention based on partial differential equation residual diagnostics. It also uses self-supervised pretraining through masked physics prediction and equation consistency prediction. The experiments are conducted on two real benchmark cases: cylinder-wake flow and fluid-structure interaction. All approaches are evaluated under a shared local protocol and compared with spectral, transformer-based, operator-learning, and physics-informed neural-network baselines. On the cylinder-wake benchmark, the proposed model achieves the best aggregate accuracy, with an all-channel normalized mean-squared error of \(0.05875\) and an all-channel Pearson correlation coefficient of \(0.97019\). On the fluid-structure-interaction benchmark, it again gives the lowest all-channel normalized mean-squared error of \(2.70\times10^{-4}\), compared with \(4.02\times10^{-4}\) for the strongest baseline. Component-wise field comparisons and scale-separated diagnostics further show stronger recovery of localized wake structures, including near-body, wake-core, and far-wake features. This indicates that the proposed model captures turbulence-associated multiscale flow behavior more effectively than the compared baselines. Overall, these results show that the proposed model improves real-world flow reconstruction while maintaining a practical accuracy-cost tradeoff.
\end{abstract}





\end{frontmatter}


\section{Introduction} \label{sec:intro}
Partial differential equations (PDEs) lie at the heart of scientific computing, governing dynamic processes across disciplines including fluid mechanics, climate modeling, materials engineering, and chemical reaction systems. A central example is computational fluid dynamics (CFD), where resolving multiscale flow systems often requires fine spatial meshes and small time steps. Many of these systems are characterized by multiscale phenomena, where solution behavior varies across widely separated spatial and temporal scales. For example, in subsurface transport, slow background flow coexists with sharp concentration fronts, while in turbulence, coherent structures span orders of magnitude in frequency and scale. Traditional numerical solvers, while highly accurate, often require prohibitively fine meshes and small time steps to resolve such behavior \cite{liu2022hierarchical}. As a result, solving multiscale PDEs with conventional methods can become computationally expensive, especially when parameter sweeps, uncertainty quantification, or real-time inference are needed.

To alleviate this burden, the field has increasingly turned to machine learning-based surrogates that aim to approximate PDE solutions with learned models. Recent engineering surrogate-modeling studies have shown that physics-informed surrogates can improve extrapolative flow prediction, while neural-operator transformers can provide high-fidelity surrogates for time-dependent nonlinear PDEs~\cite{wang2025cardiovascular,liu2026sequential}. Among these approaches, \textit{physics-informed neural networks} (PINNs) have received significant attention \cite{raissi2024physics}. By embedding the governing equations into the loss function, PINNs enforce physical laws without requiring large datasets, offering a mesh-free alternative to finite difference or finite element methods. However, while elegant in theory, vanilla PINNs face major practical challenges. Their basic architecture primarily involves a single multilayer perceptron (MLP) acting on coordinate inputs. This often limits the model's capacity to represent complex solution manifolds \cite{yi2025transforming}. Training becomes especially fragile when addressing stiff systems, chaotic attractors, or nonlocal dependencies, due to gradient pathologies and ill-conditioned loss landscapes \cite{zhang2024physics}. These issues are further compounded in multiscale settings, where the model must learn to represent both global structure and localized, high-frequency features from the same signal. Despite numerous enhancements (e.g., adaptive loss balancing \cite{rui2023reconstruction}, domain decomposition, causal training \cite{penwarden2023unified}), PINNs still struggle to scale to high-dimensional or highly nonlinear PDEs \cite{abbasi2025challenges}.

Parallel studies in \textit{neural operator learning} have proposed an alternative formulation -- rather than solving a specific PDE instance, neural operators learn mappings between function spaces, enabling generalization across a family of problems \cite{kovachki2023neural}, \cite{serrano2023operator}, \cite{raonic2023convolutional}. Models such as the Fourier Neural Operator (FNO) have demonstrated success in learning parametric solution operators, with applications in weather forecasting, fluid dynamics, and porous media flows \cite{wen2022u}. By leveraging global spectral representations, these methods achieve zero-shot super-resolution and fast inference. Nonetheless, operator-based models are largely data-driven. They often require substantial training data and do not explicitly incorporate physics beyond solution input-output pairs \cite{sinha2025effectiveness}. Furthermore, their global representations may overlook localized effects or fine-scale structures unless augmented with specialized embeddings \cite{wang2024recent}.


Against this backdrop, transformer architectures have emerged as a promising bridge between expressiveness and inductive bias. Originally developed for sequence modelling in natural language processing, transformers excel at capturing long-range dependencies through self-attention mechanisms \cite{luo2023self}. In the context of PDE learning, this property makes them well-suited for modelling interactions across distant spatial or temporal locations. This is crucial in systems where far-field behavior is influenced by localised dynamics or boundary conditions. Early applications of transformers in physics-informed learning such as PINNsFormer in 2023 \cite{zhao2023pinnsformer}, and more recently PITT \cite{Lorsung_2024}, have demonstrated improvements over MLP-based PINNs and even operator networks, particularly in time-dependent or chaotic regimes. Furthermore, recent advances in self-supervised masked pretraining for PDEs suggest that transformer-based models can learn general-purpose latent representations of physics, enabling transfer to new equations or domains \cite{wang2024recent}.

Despite these promising developments, current transformer-based PDE frameworks often remain limited in scope. Most focus on a single enhancement such as temporal attention or next token prediction. Moreover, few address the multiscale nature of PDEs head-on by explicitly modeling both local and global structures in the input field. There remains a clear gap: how can we design a transformer-based model that (i) faithfully captures multiscale solution features, (ii) respects the structure of the governing PDEs, and (iii) learns generalizable physics priors without heavy reliance on labeled data?

In this work, we introduce the Physics-Informed Bidirectional Encoder Representation Transformer, abbreviated as PIBERT, to address this challenge. Our core objective is to develop a unified framework that embeds multiscale physics knowledge directly into the architecture, training strategy, and inference pipeline of a transformer model. PIBERT is designed from the ground up to handle systems with rich multiscale behavior and complex domain geometries, aiming to serve as a high-fidelity, generalizable surrogate for PDE simulations.
To this end, our research is guided by the following questions:

\begin{itemize}
    \item \textbf{RQ1:} Can a hybrid spectral representation (combining Fourier and Wavelet embeddings) improve the model's ability to capture both global structure and fine-scale local dynamics in PDE solutions?
    \item \textbf{RQ2:} How can we incorporate the geometry and operator structure of PDEs into the transformer's attention mechanism to bias it toward physically meaningful interactions?
    \item \textbf{RQ3:} Does self-supervised pretraining on physics-inspired tasks---such as masked point prediction and edge continuity prediction---enable more robust generalization, especially in data-scarce or extrapolative regimes?
\end{itemize}

Through these questions, we seek to improve predictive accuracy and further advance the interpretability, scalability, and trustworthiness of physics-informed deep learning. The main empirical study is performed on the real cases of RealPDEBench cylinder-wake and fluid-structure-interaction-real (FSI) benchmarks \cite{hu2026realpdebench}, evaluated under shared local next-step velocity-prediction protocols with explicit provenance, multiscale diagnostics, and cost disclosure. Additional supplementary analyses on CFDBench cylinder, tube, and cavity cases~\cite{luo2023cfdbench}, a fluorocarbon ICP plasma dataset~\cite{daly2023icp}, and the EAGLE turbulent-flow dataset~\cite{janny2023eagle} are included to examine whether the same architectural behavior persists across broader flow and physics settings. These supplementary studies are used as supporting evidence for robustness and multiscale representation quality, while the primary benchmark ranking is based on the RealPDEBench protocols.

We next review recent advances in physics-informed learning and operator surrogates to situate our contribution (\Cref{sec:related}). We then introduce \textsc{PIBERT}, including its hybrid Fourier-wavelet encoder, physics-biased attention mechanism, masked physics prediction (MPP), equation consistency prediction (ECP) pretraining strategy and summarize the supporting mathematical analysis (\Cref{sec:pibert,app:extended-proofs}). Datasets, the evaluation protocol, and reproducibility details are presented in \Cref{sec:method}. We then report results on cylinder and FSI in \Cref{sec:bench}, discuss implications and limitations in \Cref{sec:implications}, and conclude in \Cref{sec:conclusion}. Benchmark provenance and local learning pipelines are documented in \Cref{sec:pde-data}. The supplementary material provides additional supporting evidence across two levels. Sections S1--S3 report broader analyses on CFDBench, ICP Plasma, EAGLE, Tube, and Cavity cases, including cross-dataset comparisons, embedding-physics alignment, and spectral diagnostics. Sections S4--S5 provide additional RealPDEBench diagnostics for cylinder-real and FSI-real, including multiscale summaries, temporal traces, timestep predictions, and additional cross-sections.

This organization addresses three issues that frequently weaken surrogate-model comparisons: contour-only evidence, ambiguous benchmark provenance, and missing cost or scale-separated analysis. Accordingly, the paper places aggregate quantitative evidence first and treats visual panels as supporting examples rather than decisive evidence. 
It also separates benchmark-source facts from the local tensorized learning protocol used by the models.

\section{Related Works}\label{sec:related}

Over the last decade, the integration of deep learning with physical modeling has become a transformative approach in scientific computing, particularly for solving complex partial differential equations (PDEs). This integration has sparked the development of a wide array of physics-informed machine learning (PIML) techniques, which have evolved in parallel with advancements in deep learning architectures, particularly neural networks, transformers, and self-supervised learning. In this section, we explore the key recent developments (2023--2025) in the field, emphasizing the challenges and innovations that have led to the creation of frameworks like PIBERT.

\subsection{Physics-Informed Neural Networks (PINNs) and Early Challenges}

PINNs, introduced by Raissi et al.~\cite{raissi2019physics} in 2019, were a groundbreaking development that integrated physics constraints directly into the loss function of neural networks to solve PDEs. These networks utilize the physics of the problem (e.g., conservation laws, boundary conditions) to inform the training process. In engineering field prediction, Roy and Guha~\cite{roy2023data} showed that physics-constrained deep learning can embed mechanical constraints through multi-objective loss terms. Despite their early success, recent reviews by Raissi et al.~\cite{raissi2024physics} and Abbasi et al.~\cite{abbasi2025challenges} highlighted several limitations of PINNs, including difficulty handling stiff equations, poor performance with shock capturing, and struggles with multiscale phenomena. Recent studies have therefore explored residual-based attention and other adaptive weighting mechanisms to improve the stability and accuracy of physics-informed training~\cite{anagnostopoulos2024residual}. These issues are partly due to the point-wise evaluation of PINNs, which make it challenging to model long-range dependencies in spatial and temporal domains. PINNs also fail to leverage the full spectrum of physical symmetries inherent in the problems they aim to solve.

\subsection{Neural Operators for Parametric PDEs}

A promising advancement beyond PINNs is the development of \textit{neural operators}, which aim to improve generalization across a wide range of physical configurations. An early milestone in operator learning is \textit{DeepONet}, which models mappings between function spaces with a branch--trunk architecture. This provided early evidence for the universal approximation of nonlinear operators from sparse sensor measurements~\cite{lu2019deeponet}. This operator-centric viewpoint set the stage for later neural-operator designs. Two recent studies have also extended this direction through vision-transformer operators and geometry-aware neural operators, showing that attention-based and geometry-conditioned architectures are increasingly central to operator learning for PDE-governed fields~\cite{ovadia2024vito,zhong2025geometry}.

The Fourier Neural Operator (FNO), introduced by Li et al.~\cite{li2020fourier}, marked a paradigm shift by learning mappings between function spaces instead of point-wise solutions, thus offering better generalization across different boundary conditions and physical parameters. FNO and its variants have seen significant enhancements, such as the U-FNO by Wen et al.~\cite{WEN2022104180}, which incorporated multiphase flow problems, and physics-embedded FNOs developed by Xu et al.~\cite{xu2024physics} that integrate physics constraints directly within the Fourier layers. These developments demonstrate improved flexibility and efficiency in solving parametric PDEs.

The move to wavelet-based methods has also been significant in dealing with localized features and discontinuities. \textit{Deep Wavelet Neural Networks} (DWNNS) introduced by Li et al.~\cite{math10121976} leverage wavelet bases for the solution of PDEs with sharp discontinuities. More recent advancements by Su et al.~\cite{Su_Ma_Tong_Xu_Chen_2024} integrated \textit{multiscale attention wavelet operators}, which proved effective in biochemical systems with steep gradients, thus broadening the scope of spectral methods in PIML. However, the trade-off between global pattern recognition (via Fourier-based methods) and local feature extraction (via wavelets) remains an ongoing challenge. Hybrid models that can balance these two domains are now a key area of research. 

Recent advances have further integrated spectral representations into deep learning frameworks for PDEs. FourierFlow addresses spectral bias in fluid dynamics through a generative framework that incorporates frequency-aware weighting and surrogate feature alignment, demonstrating improved performance in turbulence modeling \cite{wang2026fourierflow}. While primarily designed for generative forecasting, its architecture emphasizes explicit control over frequency components—a direction complementary to our hybrid spectral embedding approach. Similarly, WaveDiff leverages wavelet transforms within a diffusion-based framework to enable high-fidelity super-resolution of PDE solutions, exploiting the multi-scale localization properties of wavelets for enhanced detail recovery \cite{hu2025wavelet}. These works highlight the growing importance of incorporating domain-specific signal priors—such as scale separation and frequency structure—into neural solvers. In contrast to these methods, PIBERT unifies both Fourier and wavelet representations within a single transformer architecture, enabling simultaneous global and local field modeling, while enforcing physical consistency through physics-informed attention and self-supervised pretraining.

\subsection{Transformers in Scientific Computing}

Transformers, originally developed for natural language processing (NLP) tasks, have increasingly been adapted for scientific computing, particularly for solving PDEs and modeling long-range dependencies in physical systems. For example, Hemmasian and Barati Farimani~\cite{hemmasian2024multiscale} used transformer architectures for multi-scale PDE time-stepping, demonstrating their ability to reduce accumulated temporal error in dynamical systems. Recent work by Lorsung et al.~\cite{Lorsung_2024} also introduced the \textit{Physics-Informed Token Transformer (PITT)}, which applies self-attention mechanisms to PDE solution fields. This architecture is specifically designed to capture spatiotemporal dependencies across large datasets, demonstrating significant improvements in modeling long-range correlations in systems such as fluid dynamics and heat transfer. More recently, Liu et al.~\cite{liu2026sequential} introduced a sequential neural-operator transformer for high-fidelity surrogates of time-dependent nonlinear PDEs, reinforcing the relevance of transformer-operator architectures for engineering simulation. However, PITT and similar approaches have encountered computational challenges in scaling to high-resolution problems, with attention complexity growing quadratically.

The work of Luo et al.~\cite{luo2023self} further refined this by introducing \textit{physics-aware attention}, which modulates attention weights to respect the physical symmetries of the underlying system. This method ensures that the model better adheres to principles such as conservation laws and energy balance, thus making the model more physically interpretable and reliable. Additionally, Yang et al.~\cite{yang2022learning} explored combining autoencoders with attention mechanisms, demonstrating that enforcing physical priors (such as known dynamics or boundary conditions) within transformer architectures can substantially improve generalization in high-dimensional physical systems.

One critical insight from these works is that while transformers are highly effective in capturing long-range dependencies, they often fail to respect the inherent physical structures of scientific problems unless explicitly designed to do so. Recent function-space analysis strengthens this view by formulating attention directly as an operator between infinite-dimensional spaces, providing a principled bridge between transformer interactions and neural-operator learning \cite{calvello2025continuum}. Related operator-theoretic analysis has also shown that transformer-style architectures can be interpreted in projection-based terms, which helps explain why attention mechanisms can be meaningful for PDE operator approximation beyond sequence modeling analogies \cite{cao2021choose}. Related studies have also advocated attention mechanisms with stronger physical structure or operator-aware inductive bias, which can improve transformer performance on PDE problems by better aligning token interactions with the underlying physics \cite{wu2024transolver,calvello2025continuum}.

\subsection{Self-Supervised Learning for Physical Systems}

The application of \textit{self-supervised learning} (SSL) in physical systems is an area that has rapidly gained attention, particularly for leveraging unlabeled simulation data. In 2023, Berend et al.~\cite{berend2023masked} demonstrated that \textit{masked latent semantic modeling}, a technique inspired by SSL in NLP, could be adapted to pre-train models for physical systems. This work marked a significant shift toward unsupervised learning approaches in PIML, offering the potential to leverage abundant unlabeled data from simulations and experiments.

In 2025, Garnier et al.~\cite{garnier2025meshmask} proposed the \textit{Mesh-Mask} framework, which integrates masked graph neural networks (GNNs) for physics-based simulations. This approach was particularly effective for incomplete observational data, allowing models to learn physical consistency directly from the structure of the data. By learning the underlying patterns of physical systems without relying on labeled training data, this method improves robustness, especially when dealing with sparse or noisy datasets.

The development of masked prediction models in physical systems represents a crucial step forward, as it enables the model to generate physically consistent outputs even when labeled data is limited or unavailable. These innovations point to a new era of self-supervised pretraining for scientific machine learning, with the potential to reduce dependence on labeled datasets and broaden the practical use of simulation and experimental data.

\subsection{Benchmarking and Evaluation Frameworks}

As the PIML field matures, standardized benchmarking and evaluation frameworks have become essential for comparing different models. The CFDBench benchmark, introduced by Luo et al.~\cite{luo2023cfdbench}, remains an important synthetic testbed for spatiotemporal generalization in fluid dynamics and helped establish shared evaluation practice for surrogate models.

More recent benchmark design has moved toward paired real-world datasets with documented release protocols. RealPDEBench extends this direction by assembling complex physical systems with measured or experimentally grounded observations together with official split metadata \cite{hu2026realpdebench}. For such benchmarks, clear reporting requires a clear separation between \emph{benchmark-source facts} and the \emph{local learning representation} actually passed to the models. That distinction directly affects how modality, spatial resolution, split usage, and the availability of numerical details should be reported.

Building on this, Wang et al.~\cite{wang2024recent} emphasized that credible CFD benchmarking should evaluate both accuracy and efficiency under transparent reporting conventions. In practice, this means that convincing surrogate-model comparisons should not rely on contour agreement alone: they should disclose split provenance, target channels, scale-separated diagnostics, and any limits on cost comparability or baseline reproduction.

\subsection{Bridging Current Gaps}

Despite these advances, several critical gaps remain in the literature. First, while methods like FNO and PITT have demonstrated effectiveness in capturing global structures or long-range dependencies, they often fail to adequately capture localized features or sharp discontinuities. Second, while self-supervised learning holds promise for enhancing model robustness, few methods have been developed specifically for the unique challenges of physical systems, such as maintaining physical consistency in learned representations.

PIBERT, introduced in this work, addresses these gaps by combining Fourier and wavelet embeddings to capture both global smooth structures and sharp localized features. In addition, PIBERT incorporates physics-constrained attention mechanisms, which bias interactions toward physically meaningful patterns, ensuring that the model respects the symmetries and conservation laws inherent in the physical system. Furthermore, PIBERT’s novel self-supervised pretraining objectives, such as Masked Physics Prediction and Equation Consistency Prediction, are specifically designed to address the challenges of physics-informed learning, providing a framework that is both scalable and robust.

The central empirical question is therefore narrower and more concrete: under a shared local protocol on real-world benchmarks, does this architecture provide consistent gains in aggregate accuracy, local multiscale fidelity, and physically meaningful diagnostics, and how do those gains trade off against optimization cost relative to recent baselines?

\section{PIBERT}
\label{sec:pibert}

\subsection{Mathematical Foundations of BERT}
BERT (Bidirectional Encoder Representations from Transformers), introduced in 2019, is a language representation model based on deep bidirectional self-attention that captures contextual relationships between tokens \cite{devlin2019bert},\cite{koroteev2021bert}. Unlike unidirectional models, BERT employs a masked language model objective whereby a portion of the tokens in the input are randomly masked and the model is trained to predict the original token from its surrounding context. The architecture is based on the Transformer encoder and relies on the self-attention mechanism that is mathematically defined as
\[
\text{Attention}(\mathbf{Q}, \mathbf{K}, \mathbf{V}) = \text{softmax}\left(\frac{\mathbf{Q}\mathbf{K}^T}{\sqrt{d_k}}\right)\mathbf{V},
\]
where $\mathbf{Q}$, $\mathbf{K}$, and $\mathbf{V}$ are the query, key, and value matrices, respectively, and $d_k$ is the dimensionality of the keys. This formulation enables the model to consider both left and right context simultaneously. In addition, BERT uses a next sentence prediction task that further enhances its ability to capture relationships between sentence pairs, making it suitable for a wide range of natural language processing tasks. The overall design, which unifies pre-training and fine-tuning within the same architecture, has been critical in setting new benchmarks in language understanding.

\subsection{PIBERT Architecture}
The PIBERT architecture extends the conventional transformer-based BERT framework by incorporating physics-informed embeddings, constraints, and attention mechanisms. Unlike natural language processing models, where positional encoding is used to capture sequence order, PIBERT integrates domain knowledge directly into its embedding space, attention mechanism, and loss function to ensure that the learned representations adhere to fundamental physical principles. Figure \ref{fig:pibert_arch} shows the detailed architectural overview of PIBERT. The following sections outline the core components of PIBERT, detailing the mathematical formulations that underpin its architecture. 

\begin{figure}[htbp]
    \centering
    \includegraphics[width=\linewidth]{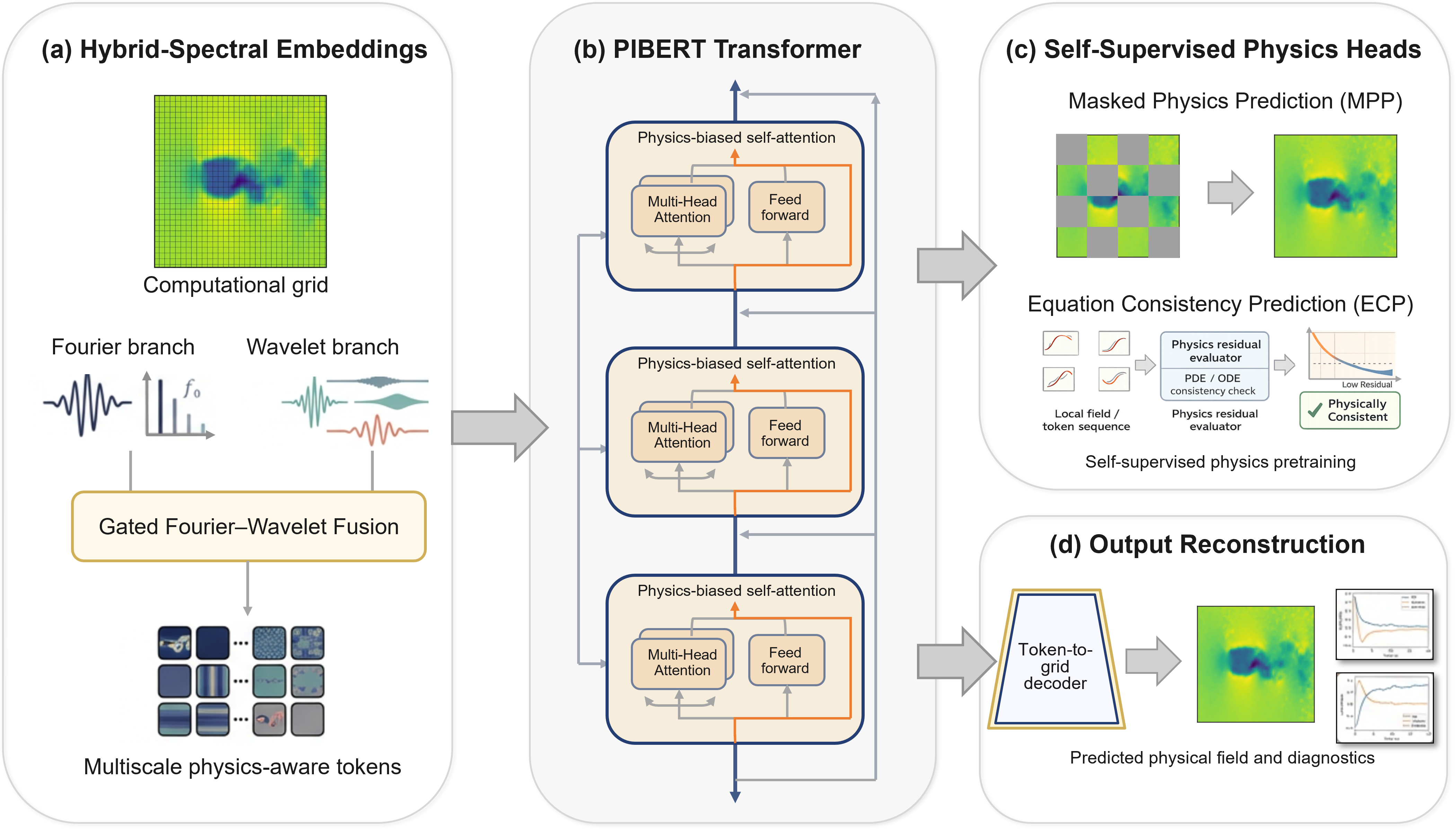}
    \caption{PIBERT architecture and training objectives.
(a) Raw CFD fields and PDE parameters defined on a spatial grid are encoded into a sequence of tokens using a gated combination of Fourier- wavelet embeddings, yielding multiscale, physics-aware token representations.
(b) The token sequence is processed by PIBERT $N$ stacked bidirectional transformer encoder.
(c) Two auxiliary heads branch from the encoder: Masked Physics Prediction (MPP), which reconstructs masked field values, and Equation Consistency Prediction (ECP), which classifies physically valid versus invalid PDE solutions with parameter adaptation.
(d) A token-to-grid decoder reconstructs the predicted physical field $\hat{y}$ and diagnostic quantities.}
    \label{fig:pibert_arch}
\end{figure}

\subsection{Physics-Informed Hybrid Spectral Embeddings}
\label{sec:hybrid-embed}

We embed grid fields with a \emph{hybrid Fourier–wavelet encoder} that provides global spectral context and local feature sensitivity. This hybridization is also motivated by recent spectral analyses showing that purely Fourier neural operators can favor dominant low-frequency content while under-representing weaker or more localized components \cite{qin2024spectral}.

\begin{table}[t]
\centering
\footnotesize
\setlength{\tabcolsep}{4pt}
\caption{Core symbols used in Sec.~\ref{sec:hybrid-embed}. Units in parentheses.}
\label{tab:symbols_core}
\begin{tabularx}{0.9\linewidth}{M >{\raggedright\arraybackslash}X}
\toprule
\textbf{Symbol} & \textbf{Meaning (units)}\\
\midrule
\Omega & Spatial domain (--)\\
(x,y),\ t & Spatial coordinates, time ($[L]$, $[T]$)\\
H,\ W & Grid height, width (px)\\
P & Patch size (px)\\
N=HW/P^{2} & Number of tokens (--)\\
u,v,\ \mathbf{u}=(u,v) & Velocity components / vector ($[L/T]$)\\
p & Pressure ($[M/(L\,T^{2})]$)\\
\rho,\ \nu & Density ($[M/L^{3}]$), kinematic viscosity ($[L^{2}/T]$)\\
\omega=\partial_x v-\partial_y u & Vorticity ($[1/T]$)\\
\nabla,\ \Delta & Gradient; Laplacian (--)\\
\alpha_F & Fourier–wavelet fusion gate (--)\\
\lambda_{\mathrm{att}} & Attention-bias strength (--)\\
Q,\ K,\ V & Attention queries, keys, values (--)\\
L_{ij},\ \alpha_{ij} & Attention logits; softmax weights (--)\\
\mathcal{L}_{\mathrm{recon}},\ \mathcal{L}_{\mathrm{phys}} & Data loss; physics penalties (--)\\
\FFT,\ \IFFT & Discrete Fourier transform and inverse (--)\\
\rfft,\ \irfft & Real 2-D FFT and inverse (--)\\
\bottomrule
\end{tabularx}
\end{table}

\subsubsection{Fourier branch (per-frequency mixing)}
Given $x\in\R^{B\times C\times H\times W}$, define $X=\mathrm{rfft2}(x)$ and apply a per-frequency channel mix on the kept half-plane; inverse rFFT returns $y_{\mathrm{ft}}\in\R^{B\times D\times H\times W}$. When the per-frequency mixing is column-unitary on the kept band, the map is nonexpansive (an isometry on band-limited inputs).

\begin{proposition}[Energy preservation of the Fourier branch]
\label{prop:fourier-energy}
If the per-frequency mixing matrices satisfy $W(h,w)^\top W(h,w)=I$ on the  modes and non-kept modes are zeroed, then the Fourier branch is $1$-Lipschitz in $\ell_2$; if inputs are band-limited to the  set, it is an isometry.
\end{proposition}

\subsubsection{Tight-frame branch (undecimated local filters)}
We use four translation-invariant filters $\{K_{LL},K_{LH},K_{HL},K_{HH}\}$ whose discrete Fourier responses form a partition of unity on the torus, yielding exact energy partition and perfect reconstruction (Parseval frame).

\begin{proposition}[Parseval tight frame, exact energy partition]
\label{prop:tight-frame}
For all $x$, $\sum_s\|K_s * x\|_2^2=\|x\|_2^2$ and $x=\sum_s K_s^\vee * (K_s*x)$; analysis and synthesis are $1$-Lipschitz.
\end{proposition}

\subsubsection{Hybrid fusion by a scalar softmax gate}
Let $E = \alpha_F \, y_{\mathrm{ft}} + (1-\alpha_F)\,y_{\mathrm{tf}}$ where
$\alpha_F = \exp(\gamma_F)/(\exp(\gamma_F)+\exp(\gamma_W)) \in [0,1]$.

\begin{lemma}[Nonexpansive hybrid for fixed gates]
\label{prop:hybrid-nonexp}
Assuming the Fourier and tight-frame branches are $1$-Lipschitz maps in $x$.
Then, for any \emph{fixed} gate field $\alpha_F$ with values in $[0,1]$ (treated
as constant with respect to perturbations in $x$), the hybrid map
$x \mapsto E(x) = \alpha_F \odot y_{\mathrm{ft}}(x) + (1-\alpha_F)\odot y_{\mathrm{tf}}(x)$
is $1$-Lipschitz in $x$. If in addition the input is band-limited and the frame
branch reduces to the identity, the hybrid is an isometry on that subspace.

In practice, $\alpha_F$ is produced by a small network that depends on $x$;
we therefore regard this result as a conditional nonexpansiveness of the
fusion step \emph{given} the gate, rather than a global Lipschitz bound
for the entire gated block.
\end{lemma}

\noindent
\textit{Proof sketches and constructions for \cref{prop:fourier-energy,prop:tight-frame,prop:hybrid-nonexp} are given in \ref{app:extended-proofs}.}

\subsubsection{Physics-Constrained Self-Attention}
\label{sec:phys-attn}

Standard attention uses logits
$L_{ij} = \langle Q_i, K_j \rangle / \sqrt{d_k}$ and
$\alpha_{ij} = \mathrm{softmax}_j(L_{ij})$.
We introduce a \emph{physics bias} by subtracting a nonnegative proxy
$R_{ij} \ge 0$ derived from PDE diagnostics,
\begin{equation}
\tilde L_{ij}
\;=\; L_{ij} - \lambda_{\mathrm{att}}\,R_{ij},
\qquad
\alpha_{ij}(\lambda_{\mathrm{att}})
\;=\;
\frac{\exp(\tilde L_{ij})}{\sum_m \exp(\tilde L_{im})},
\label{eq:biased-attn}
\end{equation}
By subtracting localized PDE residuals from the attention logits, we penalize attention directed toward tokens associated with larger physical violations. This encourages the transformer to rely more strongly on physically consistent latent features during bidirectional information mixing.

In general, $R_{ij}$ can be pairwise, for example by depending on relative position, geometry, or the physical states associated with both tokens. In this work, we use the following separable, key-dependent instance: \begin{equation} R_{ij} \;=\; r_j, \qquad r_j \ge 0, \label{eq:keywise-residual-bias} \end{equation} where $r_j$ is the pooled physical residual associated with the attended key token $j$. The resulting bias is $P_{ij}=-\lambda_{\mathrm{att}}r_j$ and is broadcast across all query positions and attention heads. Thus, an explicit dense pairwise residual matrix does not need to be stored. The residual $r_j$ is obtained by pooling grid-level diagnostics, such as divergence and momentum-style proxies, over the receptive-field support of token $j$, yielding the \texttt{R\_tok} vector used in Algorithm~\ref{alg:pibert-compact}. A query-only shift $R_{ij}=r_i$ is not used because it is constant across all keys within row $i$ and therefore cancels exactly under rowwise softmax: \begin{equation} \operatorname{softmax}_{j}\!\left(L_{ij} -\lambda_{\mathrm{att}}r_i\right) = \operatorname{softmax}_{j}\!\left(L_{ij}\right). \end{equation} Consequently, the operative residual bias must vary with the key index $j$, or more generally with both $i$ and $j$, in order to modify the relative attention weights.

\begin{lemma}[Softmax ratio monotonicity]
\label{lem:ratio}
For fixed row $i$, $\alpha_{ij_1}/\alpha_{ij_2}=\exp\!\big((L_{ij_1}-L_{ij_2})-\lambda_{\mathrm{att}}(R_{ij_1}-R_{ij_2})\big)$; hence if $R_{ij_1}>R_{ij_2}$ the ratio decreases monotonically with $\lambda_{\mathrm{att}}$.
\end{lemma}

\begin{lemma}[Rowwise Lipschitz control]
\label{lem:lipschitz}
Let $\alpha_i(\lambda)$ denote row-$i$ weights under \eqref{eq:biased-attn}. Then
$\|\alpha_i(\lambda)-\alpha_i(0)\|_1 \le \frac{\lambda_{\mathrm{att}}}{2}\,\|R_{i\cdot}\|_\infty$.
\end{lemma}

\begin{proposition}[Translation equivariance on the torus]
\label{prop:equivariance}
Index tokens by lattice sites $r(i)\in\mathbb{Z}_H\times\mathbb{Z}_W$ (periodic). If $R_{ij}=\rho\!\big(p,\;r(i)-r(j)\big)$ depends only on relative position and parameters, then attention with bias \eqref{eq:biased-attn} is translation-equivariant.
\end{proposition}

\begin{theorem}[Continuum (kernel) limit]
\label{prop:continuum}
Placing tokens on a regular grid with spacing $h\to 0$ and using Riemann sums, the biased attention converges uniformly to a nonlocal kernel operator
$(Tv)(x)=\int_\Omega w_\lambda(x,y)v(y)\,dy$ with
$w_\lambda(x,y)\propto \exp(\langle q(x),k(y)\rangle-\lambda r(x,y))$.
\end{theorem}

\noindent
\textit{Proofs for \cref{lem:ratio,lem:lipschitz,prop:equivariance,prop:continuum} are given in the \ref{app:extended-proofs}}
\paragraph{Instantiation for incompressible Navier–Stokes}
With velocity $u=(u,v)$, pressure $p$, density $\rho$, viscosity $\nu$, define diagnostics
\[
R^{(\mathrm{div})}=\big|\nabla\cdot u\big|,\quad
R^{(\mathrm{mom})}=\left\|
\partial_t u + (u\cdot\nabla)u + \frac{1}{\rho}\nabla p - \nu\,\Delta u
\right\|_2,\quad
R=\alpha_{\mathrm{div}}R^{(\mathrm{div})}+\alpha_{\mathrm{mom}}R^{(\mathrm{mom})}.
\]
We evaluate $R$ on the grid via central differences and pool it to the token level. Only the token-level residual vector is stored and broadcast across query positions, preserving the $O(BN^2d)$ cost of attention. All theoretical statements above hold for general $R_{ij}$; our actual implementation corresponds to the separable special case $R_{ij}=r_j$, i.e.\ a key-dependent scalar bias derived from pooled PDE diagnostics. In this paper, that bias is instantiated for incompressible-flow diagnostics. For other PDE families, the same architecture can still be used, but the design of $R_{ij}$ must be rebuilt around the appropriate residuals, invariants, and boundary constraints of the target system.

\subsubsection{Self-Supervised Objectives (MPP/ECP) and Physics Coupling}
\label{sec:ss-objectives-short}
With mask $M\in\{0,1\}^{H\times W}$ and input $\tilde{x}=M\odot x$, the MPP loss is
\begin{equation}
\mathcal{L}_{\mathrm{mpp}}=\frac{1}{|\{(i,j):M_{ij}=0\}|}\sum_{M_{ij}=0}\|f_\theta(\tilde{x},p)_{ij}-x_{ij}\|_2^2.
\label{eq:mpp}
\end{equation}
\begin{proposition}[Population minimizer]
\label{prop:mpp-condexp}
Under MSE risk, the population minimizer is the conditional mean $f^\star(\tilde{x},p)=\mathbb{E}[x\mid \tilde{x},p]$.
\end{proposition}

\paragraph{Divergence-aware regularization for incompressible flows}
Define the penalized population risk $\mathcal{R}_\lambda(g)=\mathbb{E}\|x-g(\tilde{x})\|_2^2+\lambda\,\mathbb{E}\|D\,g(\tilde{x})\|_2^2$ where $D$ is the discrete divergence.

\begin{theorem}[Oracle inequality toward the solenoidal class]
\label{prop:oracle-div}
If $\mathcal{H}=\ker(D)$ is the discrete divergence-free subspace and $\mathrm{dist}(u,\mathcal{H})\le c_H\|Du\|_2$, then any minimizer $g_\lambda$ of $\mathcal{R}_\lambda$ satisfies
\[
\mathbb{E}\|x-g_\lambda(\tilde{x})\|_2^2
\le
\mathbb{E}\|x-g^\star(\tilde{x})\|_2^2
+\frac{c_H^2}{\lambda}\,\mathbb{E}\|D\,g_\lambda(\tilde{x})\|_2^2
\]
for all $g^\star$ with $Dg^\star\equiv 0$.
\end{theorem}

\noindent

In this section, we have established the theoretical foundations of PIBERT, providing rigorous derivations of its physics-informed embeddings, attention mechanism, and loss function. By integrating Fourier and wavelet embeddings, enforcing physics-based constraints within self-attention, and minimizing PDE residuals in the loss function, PIBERT represents a significant advancement in transformer-based scientific modeling.

\subsubsection{Physics-Inspired Regularization Terms}
\label{sec:loss-essentials}
We combine a data fidelity loss with \emph{physics-inspired regularizers} that
bias the model toward incompressibility and smoothness, rather than enforcing
the full Navier–Stokes momentum residual explicitly:

\begin{equation}
\mathcal{L}_{\mathrm{recon}}
= \frac{1}{|\Omega|}
  \sum_{(i,j)\in\Omega}
  \big\|\hat{y}_{ij} - y_{\mathrm{true},ij}\big\|_2^2,
\label{eq:recon-loss}
\end{equation}

\begin{align}
\mathcal{L}_{\mathrm{div}}
&= \frac{1}{|\Omega|}\sum_{(i,j)\in\Omega}
   \big(\nabla\!\cdot\!\hat{\mathbf u}\big)_{ij}^{2}, \\
\mathcal{L}_{\mathrm{lap}}
&= \frac{1}{|\Omega|}\sum_{(i,j)\in\Omega}
   \Big(\big\|\Delta \hat{u}\big\|_{2}^{2}
       +\big\|\Delta \hat{v}\big\|_{2}^{2}\Big)_{ij}, \\[2pt]
\mathcal{L}_{\mathrm{reg}}
&= \lambda_{\mathrm{div}}\,\mathcal{L}_{\mathrm{div}}
 +  \lambda_{\mathrm{lap}}\,\mathcal{L}_{\mathrm{lap}} .
\label{eq:phys-loss}
\end{align}

\begin{equation}
\mathcal{L}_{\mathrm{bnd}}
= \frac{1}{|M|}\sum_{(i,j)\in M}
  \big\|\hat{\mathbf u}_{ij} - \mathbf u_{\mathrm{true},ij}\big\|_2^2 .
\end{equation}

\begin{equation}
\mathcal{L}_{\mathrm{total}}
= \mathcal{L}_{\mathrm{recon}}
+ \lambda_{\mathrm{reg}}\,\mathcal{L}_{\mathrm{reg}}
+ \lambda_{\mathrm{bnd}}\,\mathcal{L}_{\mathrm{bnd}} .
\end{equation}

Discrete Green/sum-by-parts identities used for analyzing these regularizers are stated in \Cref{lem:green-app}. A standard quadratic boundary penalty enforces Dirichlet data in the $\mu\to\infty$ limit; see \Cref{prop:dirichlet-penalty}. We emphasize that, in the current implementation, these terms act as physics-guided regularization rather than a full Navier–Stokes residual penalty. 

\subsection{Masked Physics Prediction (MPP)}

The Masked Physics Prediction (MPP) task in PIBERT is inspired by the masked language modeling (MLM) approach used in BERT, but adapted to the physics domain, where missing field values must adhere to governing physical laws. In NLP, MLM randomly masks a fraction of input tokens, and the model is trained to reconstruct them using contextual information \cite{berend2023masked}. In physics-informed learning, however, missing field values cannot be arbitrarily inferred based solely on data correlations; instead, they must conform to differential equations and boundary conditions that govern physical systems \cite{hao2022physics}. PIBERT extends the MLM concept by randomly masking portions of a continuous physical field, such as velocity, pressure, or temperature, and requiring the model to infer these missing values in a way that respects the underlying physics.

The motivation for MPP is to encourage PIBERT to develop embeddings that capture both local and global physical dependencies. Unlike PINNs and standard PDE solvers that require direct access to governing equations at all points, PIBERT learns to fill in missing physics values by leveraging self-attention mechanisms, which propagate information across spatial-temporal domains. This results in a model that generalizes better across different boundary conditions and PDE structures. 

To ensure that PIBERT does not simply interpolate missing values based on statistical patterns, but rather learns to respect physical constraints, the masking process is carefully structured. Instead of uniformly dropping values, PIBERT applies a structured masking scheme where missing values are informed by physics constraints. In this approach, $80\%$ of the masked values are completely removed from the input, forcing the model to reconstruct them solely from its learned representations. Another $10\%$ of the masked values are replaced with random noise drawn from a physics-aware distribution, challenging the model to denoise and enforce physically consistent predictions. The remaining $10\%$ of the masked values are left unchanged, ensuring that PIBERT remains aware of absolute field values and does not learn to ignore known information.

A natural question arises: why is random masking an effective strategy in physics-informed learning? In conventional PDE solvers, missing values are typically interpolated using explicit numerical schemes, while in generative models such as variational autoencoders (VAEs) \cite{yang2022learning} and physics-informed GANs, missing data is imputed via sampling from a learned latent space \cite{garnier2025meshmask},  \cite{zhou2024masked},  \cite{taghizadeh2024multi}. PIBERT takes a different approach—it does not explicitly enforce numerical interpolation but instead learns physics-aware embeddings through self-attention, leveraging long-range dependencies across a field. This enables PIBERT to capture the fundamental physics of the system without requiring explicit PDE constraints during inference.

\subsection{Equation Consistency Prediction (ECP)}

In addition to reconstructing missing field values, PIBERT is pre-trained to ensure that its learned representations comply with the governing equations of physical systems. This is achieved through Equation Consistency Prediction (ECP), a self-supervised classification task designed to reinforce physical validity within the model’s learned embeddings.

In most PDE-driven physical processes, solutions must satisfy strict mathematical constraints, including conservation laws, balance equations, and boundary conditions. Traditional solvers explicitly enforce these constraints, while PINNs incorporate them as soft constraints in the loss function \cite{zhang2024physics}. PIBERT, however, learns an implicit understanding of these constraints by classifying whether a given physics field satisfies its corresponding governing equation. This enables the model to internalize the difference between physically plausible and non-physical solutions, improving robustness and generalization.

Mathematically, let $\mathcal{N}(u)$ be the differential operator that governs a system, such that a valid solution must satisfy:

\begin{equation}
    \mathcal{N}(u) \approx 0
\end{equation}

where $\mathcal{N}(u)$ could represent equations such as:
\begin{equation}
    \frac{\partial u}{\partial t} + u \cdot \nabla u + \frac{1}{\rho} \nabla p - \nu \nabla^2 u = 0 \quad \text{(Navier-Stokes)}
\end{equation}
\begin{equation}
    \frac{\partial u}{\partial t} - \alpha \nabla^2 u = 0 \quad \text{(Heat Equation)}
\end{equation}
\begin{equation}
    \frac{\partial^2 u}{\partial t^2} - c^2 \nabla^2 u = 0 \quad \text{(Wave Equation)}
\end{equation}

To construct a dataset for training ECP, we generate solution pairs $(u_{\text{valid}}, u_{\text{invalid}})$. The valid solutions are obtained from numerical solvers that exactly satisfy the governing equations, while invalid solutions are generated by perturbing valid solutions through random noise, incorrect boundary conditions, or omitted PDE terms. PIBERT is then trained as a binary classifier, minimizing the equation consistency loss:

\begin{equation}
    \mathcal{L}_{\text{ECP}} = -\sum_{i}\Big( y_i \log \hat{y}_i + (1-y_i)\log(1-\hat{y}_i)\Big).
\end{equation}

where $y_i = 1$ if the sample is a valid physics solution, and $y_i = 0$ otherwise.

A key challenge in designing ECP is ensuring that the incorrect PDE solutions are sufficiently realistic. If the incorrect solutions are overly simplistic (e.g., random noise), PIBERT may simply learn to classify based on superficial artifacts rather than understanding true physics consistency. To mitigate this, PIBERT incorporates a gradient-based adversarial perturbation strategy \cite{xue2023diffusion}, where incorrect solutions are generated by making minimal modifications that still violate the PDE constraints. This forces PIBERT to learn deep physics-informed features, rather than relying on simple pattern recognition.

Unlike conventional regression-based loss functions, which directly minimize deviations from known PDE solutions, ECP provides an additional layer of self-supervised validation that is particularly useful when working with sparse or incomplete physics datasets. PIBERT learns not only to reconstruct missing values but also to ensure that its predictions remain physically plausible.

\subsection{PIBERT algorithm and complexity}
\label{sec:pibert-alg-complex}

PIBERT embeds grid fields with a hybrid spectral encoder (Fourier branch + tight-frame branch), fuses them via a softmax gate, projects to the model width, tokenizes (with optional parameter token), and applies $L$ transformer encoder layers with physics-biased attention. A lightweight decoder maps tokens back to grids, training minimizes a data term plus physics regularizers and (optionally) MPP/ECP. Algorithm \ref{alg:pibert-compact} provides a detailed implementation guide on training and preparing the model from scratch.

Now let $H{\times}W$ be the grid, $P$ the patch size so $N\!\approx\!HW/P^2$ tokens, width $d$, FFN width $d_{\mathrm{ff}}\!\approx\!4d$, and $D$ the pre-token feature channels. Per layer, the transformer dominates for moderate $N$:
\[
\underbrace{O\!\big(BN^2 d\big)}_{\text{self-attn}} \;+\; \underbrace{O\!\big(BN d^2 + BN d\,d_{\mathrm{ff}}\big)}_{\text{QKV/out + MLP}}.
\]
The hybrid encoder adds
\[
\underbrace{O\!\big(BC\,HW\log(HW) + BD\,HW\log(HW)\big)}_{\text{rFFT/irFFT}} \;+\; \underbrace{O\!\big(BC\,s\,k^2 HW\big)}_{\text{tight-frame filters}} \;+\; \underbrace{O\!\big(BHW\,D\,d\big)}_{\text{$1{\times}1$ proj}},
\]
with per-frequency mixing $O(BmCD)$ negligible when the kept spectral area $m\ll HW$. Computing the residual proxy on-grid and pooling to tokens is $O(BHW)+O(BN)$ and is small vs.\ attention. Peak activation memory is $O(BN^2)$ for vanilla attention (or $O(BN d)$ with a memory-efficient kernel); hybrid-encoder activations are $O(B(D{+}C)HW)$ \footnote{Dominant cost: for large $N$ (small $P$), attention $O(LBN^2 d)$ dominates. 
The Fourier branch dominates only when $N$ is small, the kept band $m\!\approx\!HW$, and $C,D$ are large. 
The tight-frame cost is linear in $C$ and $HW$ for fixed $s,k$, and the $1{\times}1$ projection is comparable to or cheaper than one attention/MLP pass.}.

\begin{algorithm}[htbp]
\caption{PIBERT training step algorithm}
\label{alg:pibert-compact}
\begin{algorithmic}[1]
\Require Batch $x\!\in\!\mathbb{R}^{B\times C\times H\times W}$; params $p\!\in\!\mathbb{R}^q$; (optional) targets $y_{\text{true}}$.
\Statex \hspace{\algorithmicindent} Patch size $P$ ($N{=}HW/P^2$); encoder width $d$; layers $L$; bias $\lambda_{\text{att}}$.
\Ensure Prediction $\hat{y}$, total loss $\mathcal{L}_{\text{total}}$.
\smallskip
\State \textbf{Hybrid spectral encoding:}
       $F \leftarrow \textproc{FourierEncode}(x)$;\quad
       $W \leftarrow \textproc{WaveletFrame}(x)$;
\State \textbf{Fuse \& project:}
       $E \leftarrow \textproc{Fuse}(F,W;\alpha_F)$ (\Cref{sec:hybrid-embed});\;
       $Z \leftarrow \textproc{Conv}_{1\times 1}(E)\in\mathbb{R}^{B\times d\times H\times W}$;
\State \textbf{Tokenize \& condition:}
       $X^{(0)} \leftarrow \textproc{Patchify}(Z,P)$;\quad
       $X^{(0)} \leftarrow \textproc{AppendParamToken}(X^{(0)},p)$;
\State \textbf{Physics bias (once per step or periodically):}
       $R_{\text{grid}} \leftarrow \textproc{PhysicsResidual}(x)$ \ (\Cref{sec:phys-attn});
       $R_{\text{tok}} \leftarrow \textproc{Pool}_{P\times P}(R_{\text{grid}})$;
       $P^{(1)} \leftarrow -\lambda_{\text{att}}\,R_{\text{tok}}$;
\For{$\ell=1$ to $L$} \Comment{Transformer encoder with physics-biased attention}
  \State $Q,K,V \leftarrow \textproc{QKV}(\textproc{LN}(X^{(\ell-1)}))$;
  \State $\alpha \leftarrow \textproc{Softmax}\!\Big(\frac{QK^\top}{\sqrt{d_k}} + P^{(\ell)}\Big); $\quad
         $H \leftarrow \alpha V$;
  \State $X' \leftarrow X^{(\ell-1)} + H W_O$;\quad
         $X^{(\ell)} \leftarrow X' + \textproc{MLP}(\textproc{LN}(X'))$;
  \State  refresh $P^{(\ell+1)} \leftarrow -\lambda_{\text{att}}\,R_{\text{tok}}$;
\EndFor
\State \textbf{Decode to grid:} $\hat{y} \leftarrow \textproc{Decode}(X^{(L)})$;
\State \textbf{Losses:}
  $\mathcal{L}_{\text{sup}}$ from \Cref{sec:loss-essentials} 
  (uses \Cref{eq:recon-loss,eq:phys-loss}); \ 
  $\mathcal{L}_{\text{mpp}}$ from \Cref{eq:mpp};\ 
  $\mathcal{L}_{\text{ecp}}$ from ECP (\Cref{sec:ss-objectives-short});
\State \textbf{Total \& update:}
  $\mathcal{L}_{\text{total}}
  = \mathcal{L}_{\text{sup}}\ (\text{if }y_{\text{true}}) 
  + \lambda_{\text{mpp}}\mathcal{L}_{\text{mpp}}
  + \lambda_{\text{ecp}}\mathcal{L}_{\text{ecp}}$; \quad
  $\theta \leftarrow \textproc{OptStep}(\nabla_\theta \mathcal{L}_{\text{total}})$;
\smallskip
\Statex\emph{Notes:} $R_{\text{grid}}$ uses divergence/momentum diagnostics; token bias is the pooled residual. 
\end{algorithmic}
\end{algorithm}

\section{Methodology}
\label{sec:method}
We aim to learn scalable and physics-faithful computational surrogates for multiscale mechanics-governed flow fields under sparse supervision. PIBERT addresses this goal through three components. First, it uses a hybrid Fourier-wavelet spectral encoder to capture both global structure and localized features. Second, it uses physics-biased self-attention based on PDE residual diagnostics to favor physically meaningful interactions. Third, it uses self-supervised pretraining through Masked Physics Prediction (MPP) and Equation Consistency Prediction (ECP) before task-specific fine-tuning. Together, these components improve data efficiency, stability, and multiscale field reconstruction.

\subsection{Datasets}
\paragraph{Benchmark-source overview.}
This empirical study uses the RealPDEBench releases \cite{hu2026realpdebench}. We study exactly two official real-world benchmarks: cylinder-real and FSI-real. Throughout the paper, benchmark-source facts reported by the official release are kept separate from the local learning representation used by the experiments, so that modality, resolution, split protocol, and preprocessing are not conflated.

\paragraph{Cylinder-real}
RealPDEBench cylinder-real packages bluff-body wake measurements built from real PIV velocity observations and associated benchmark metadata \cite{hu2026realpdebench}. At the benchmark-source level, the release reports real PIV $(u,v)$ observations at $128\times 256$, paired simulated $(u,v,p)$ fields at $64\times 128$, $92$ trajectories, $23{,}990$ frames, $400$~Hz sampling, $20$~s duration, and Reynolds numbers from $1800$ to $12000$. In the local data copy used for learning, the released real velocity stream decodes to tensors of shape $(T,2,64,128)$ per trajectory. We treat that decoded representation as the start of the learning pipeline, sample it every 20 native steps, resize each sampled frame to $(2,64,64)$, and train all compared models on the same next-step real-velocity prediction task. The resulting trajectory-level split is $73/9/10$ for train/validation/test, yielding $14{,}527/1{,}791/1{,}990$ paired samples.

\paragraph{FSI-real}
The FSI benchmark is the official real split of the tandem-cylinder vortex-induced-vibration setting from RealPDEBench \cite{hu2026realpdebench}. The source release provides benchmark fields at $128\times 128$ together with released physical metadata for the tandem-cylinder configuration and vibration setting. In the local learning pipeline, we use the released real-only velocity channels $(u,v)$, resize them to $64\times 64$, and evaluate all six models on the same next-step tensorized prediction task.

\paragraph{Source benchmark versus local learning representation}
For both benchmarks, the manuscript reports benchmark-source numerical provenance only when it is explicitly provided by the official release. The local learning problem is then stated separately in terms of tensor decoding, channel selection, resizing, sampling stride, and split protocol. We do not infer solver order, mesh order, or time-integration details when those quantities are not specified by the source benchmark documentation. This distinction is necessary for clarity and for separating benchmark-source facts from the local learning pipeline used in this study.
The study-wide benchmark inventory is summarized in \Cref{tab:benchmark-summary}, and the benchmark-source facts together with the explicit provenance limits are itemized in \Cref{tab:benchmark-source-details}.

\begin{table}[h]
\centering
\caption{Benchmark summary for the two RealPDEBench benchmarks and the local learning representation used in this study.}
\label{tab:benchmark-summary}
\resizebox{\linewidth}{!}{%
\begin{tabular}{p{2.2cm}p{3.1cm}p{3.0cm}p{2.0cm}p{4.0cm}p{3.1cm}}
\toprule
\textbf{Benchmark} & \textbf{Physical scenario} & \textbf{Source modalities} & \textbf{Source resolution} & \textbf{Local learning representation} & \textbf{Split summary} \\
\midrule
Cylinder-real & Real bluff-body wake behind a cylinder & Real PIV $(u,v)$ plus paired simulated $(u,v,p)$ in the official benchmark package & Real $128\times 256$; simulated $64\times 128$ & Released real velocity tensors decode to $64\times 128$, are sampled every 20 native steps, resized to $64\times 64$, and used for one-step $(u,v)$ prediction & Seed-42 trajectory split $73/9/10$, yielding $14{,}527/1{,}791/1{,}990$ sampled pairs \\
FSI-real & Official tandem-cylinder vortex-induced-vibration benchmark with real observations & Official real split with released flow fields and physical metadata & $128\times 128$ & Real-only $(u,v)$ channels resized to $64\times 64$ for the shared six-model comparison & Official real split reused as released, with the same tensorized protocol across all models \\
\bottomrule
\end{tabular}}
\end{table}

\begin{table}[h]
\centering
\caption{Benchmark-source facts used in the manuscript and the provenance limits kept explicit.}
\label{tab:benchmark-source-details}
\resizebox{\linewidth}{!}{%
\begin{tabular}{p{2.2cm}p{3.7cm}p{2.6cm}p{2.1cm}p{3.0cm}p{3.8cm}}
\toprule
\textbf{Benchmark} & \textbf{Operating parameters} & \textbf{Duration / frequency} & \textbf{Source grid} & \textbf{Available channels} & \textbf{Released metadata / provenance note} \\
\midrule
Cylinder-real & Reynolds-number range $1800$--$12000$ & $20$~s at $400$~Hz & Real $128\times 256$; simulated $64\times 128$ & Real $(u,v)$; simulated $(u,v,p)$ & The official release documents $92$ trajectories and $23{,}990$ frames. The local experiments use the released real velocity tensor only; solver order and mesh order are not reported by the source benchmark. \\
FSI-real & Tandem-cylinder vortex-induced-vibration setting; released physical-parameter metadata accompanies the benchmark & Official real split as released & $128\times 128$ & Released real flow fields; local study uses real-only $(u,v)$ & The benchmark release provides the governing physical scenario and split metadata. Missing solver, mesh, or numerical-scheme details are not reported here unless they appear in the official release. \\
\bottomrule
\end{tabular}}
\end{table}

\subsection{Model selection and Training}

The active comparison set is exactly six models: PIBERT, FourierFlow, FNO2d, PITT, DeepONet2d, and PINN. These models span physics-informed networks, spectral operators, and recent transformer-style surrogates. PIBERT keeps the same core architecture described in earlier \cref{sec:pibert} we do not modify the theorem/proof or algorithmic parts of the method for this study.

All compared models consume the local real-velocity tensor at time $t$ and predict the next sampled real-velocity tensor at time $t+1$ within the learning sequence. PIBERT uses the MPP/ECP-pretrained encoder followed by supervised fine-tuning with the physics-guided losses described in \Cref{sec:loss-essentials}. For the  benchmark-accuracy summaries, we use one complete run per model from the finalized comparison manifests under the same split and tensorized protocol, including PITT and FourierFlow. This supports a protocol-matched accuracy comparison, but it does \emph{not} justify a blanket claim that every reported auxiliary quantity comes from identical optimization budgets or identical training histories. Accordingly, the appendix cost table is reported only as a compact disclosure of model size and observed runtime.

The logged FSI validation convergence is shown in \Cref{fig:training-convergence}. These curves provide optimization context only and the reported results below remain tied to held-out test metrics.

\begin{figure}[htbp]
    \centering
    \includegraphics[width=\linewidth]{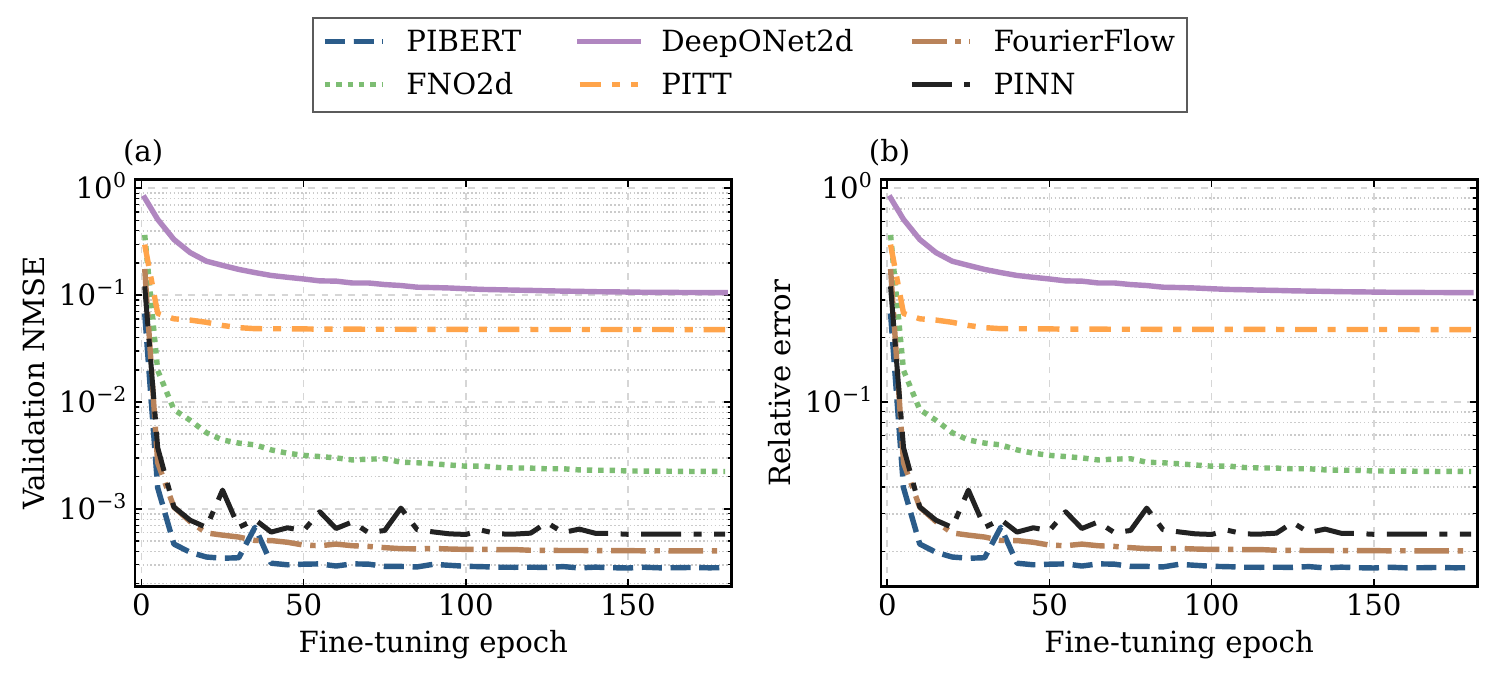}
    \caption{FSI-real validation convergence for the  comparison runs. Panel (a) shows the logged fine-tuning validation NMSE for all models on FSI Benchmark. Panel (b) reports the corresponding relative-error view from the same validation checkpoints.}
    \label{fig:training-convergence}
\end{figure}

\subsection{Evaluation protocol}

All comparisons reported use fixed train/validation/test partitions. For cylinder-real, this means the seed-42 trajectory-level split described above; for FSI-real, this means the official real split distributed with the benchmark. The shared supervised target is the real velocity field $(u,v)$ at the next sampled time step. Derived quantities such as speed magnitude $|\mathbf{u}|$ and vorticity $\omega$ are diagnostic views only and are not trained as separate targets. For cylinder-real, the local protocol samples every 20 native time steps, uses up to 200 frames per trajectory, and resizes each sampled frame from $64\times 128$ to $64\times 64$. For FSI-real, the local protocol uses the official real split, the released real-only $(u,v)$ fields, and the same $64\times 64$ learning representation across all six models.

We report four dataset-wide metrics in this paper: local mean absolute error
(LMAE), local Pearson correlation coefficient (LPCC), coefficient of determination
(\(R^2\)), and normalized mean-squared error (NMSE). The definitions used in this study are summarized in \Cref{tab:evaluation-metrics}. The reported local diagnostics include component-wise field panels, slice-based relative $\ell_2$ errors in the near-body, wake-core, and far-wake regions, and the FSI scale-separated summary. The detailed FSI optimization-cost table is reported in the appendix. The supplementary material is divided into broader supporting analyses and RealPDEBench-specific diagnostics. Sections S1--S3 examine additional datasets and diagnostic settings, including CFDBench, ICP Plasma, EAGLE, Tube, and Cavity cases. Sections S4--S5 provide additional RealPDEBench results and diagnostics for cylinder-real and FSI-real, including temporal traces, timestep predictions, multiscale summaries, and additional cross-sections. Aggregate figures and tables in this paper are dataset-wide. The field panels show deterministically selected held-out samples. The multiscale panels also use deterministic selection from the multiscale summary. These selections are reported explicitly so that the visual evidence remains auditable and is not used as the sole basis for any quantitative claim.

\begin{table}[htbp]
\centering
\footnotesize
\setlength{\tabcolsep}{3pt}
\renewcommand{\arraystretch}{1.45}
\caption{Evaluation metrics used for dataset-wide comparison. Here \(y_i\) and
\(\hat{y}_i\) denote the ground-truth and predicted values over the same test
index set \(\mathcal{I}\), \(\bar{y}\) and \(\bar{\hat{y}}\) denote their means,
all summations are over \(i\in\mathcal{I}\), and \(\varepsilon\) is a small
numerical-stability constant.}
\label{tab:evaluation-metrics}
\begin{tabularx}{\linewidth}{
p{1.5cm}
>{\centering\arraybackslash}p{5.8cm}
>{\raggedright\arraybackslash}X
p{2.2cm}}
\toprule
\textbf{Metric} & \textbf{Definition} & \textbf{Purpose} & \textbf{Direction} \\
\midrule

LMAE &
\(\displaystyle
\frac{1}{|\mathcal{I}|}\sum_i |\hat{y}_i-y_i|
\) &
Mean absolute prediction error over the evaluated field values &
Lower is better \\

\midrule

LPCC &
\(\displaystyle
\frac{\sum_i(\hat{y}_i-\bar{\hat{y}})(y_i-\bar{y})}
{\sqrt{\sum_i(\hat{y}_i-\bar{\hat{y}})^2}
\sqrt{\sum_i(y_i-\bar{y})^2}+\varepsilon}
\) &
Linear agreement between predicted and true field patterns &
Higher is better \\

\midrule

\(R^2\) &
\(\displaystyle
1-
\frac{\sum_i(\hat{y}_i-y_i)^2}
{\sum_i(y_i-\bar{y})^2+\varepsilon}
\) &
Explained variance relative to the test-set mean &
Higher is better \\

\midrule

NMSE &
\(\displaystyle
\frac{\sum_i(\hat{y}_i-y_i)^2}
{\sum_i y_i^2+\varepsilon}
\) &
Scale-normalized squared reconstruction error &
Lower is better \\

\bottomrule
\end{tabularx}
\end{table}

\subsection{Reproducibility}
The cylinder-real experiments use the seed-42 split described above, and the FSI-real experiments use the official real split. The reported runs were produced primarily with mixed precision on an NVIDIA RTX GPU with 24\,GB VRAM and a 12th-generation Intel Core i9 CPU. An A100-based local server, also paired with a 12th-generation Intel Core i9 CPU, was used only for auxiliary development checks and is not part of the reported results. The source code, preprocessing scripts, model configurations, and checkpoint records are maintained in the project repository and can be shared for review and reproduction purposes upon request.

\section{Benchmark and Results}
\label{sec:bench}
The empirical study is organized hierarchically: aggregate metrics first, then selected component-wise fields, then scale-separated and local diagnostics. This ordering is deliberate. Aggregate figures and tables carry the primary quantitative claims, while the field panels serve as interpretable examples of the same ranking.

\subsection{Cylinder-real benchmark}

\begin{figure}[htbp]
    \centering
    \includegraphics[width=\linewidth]{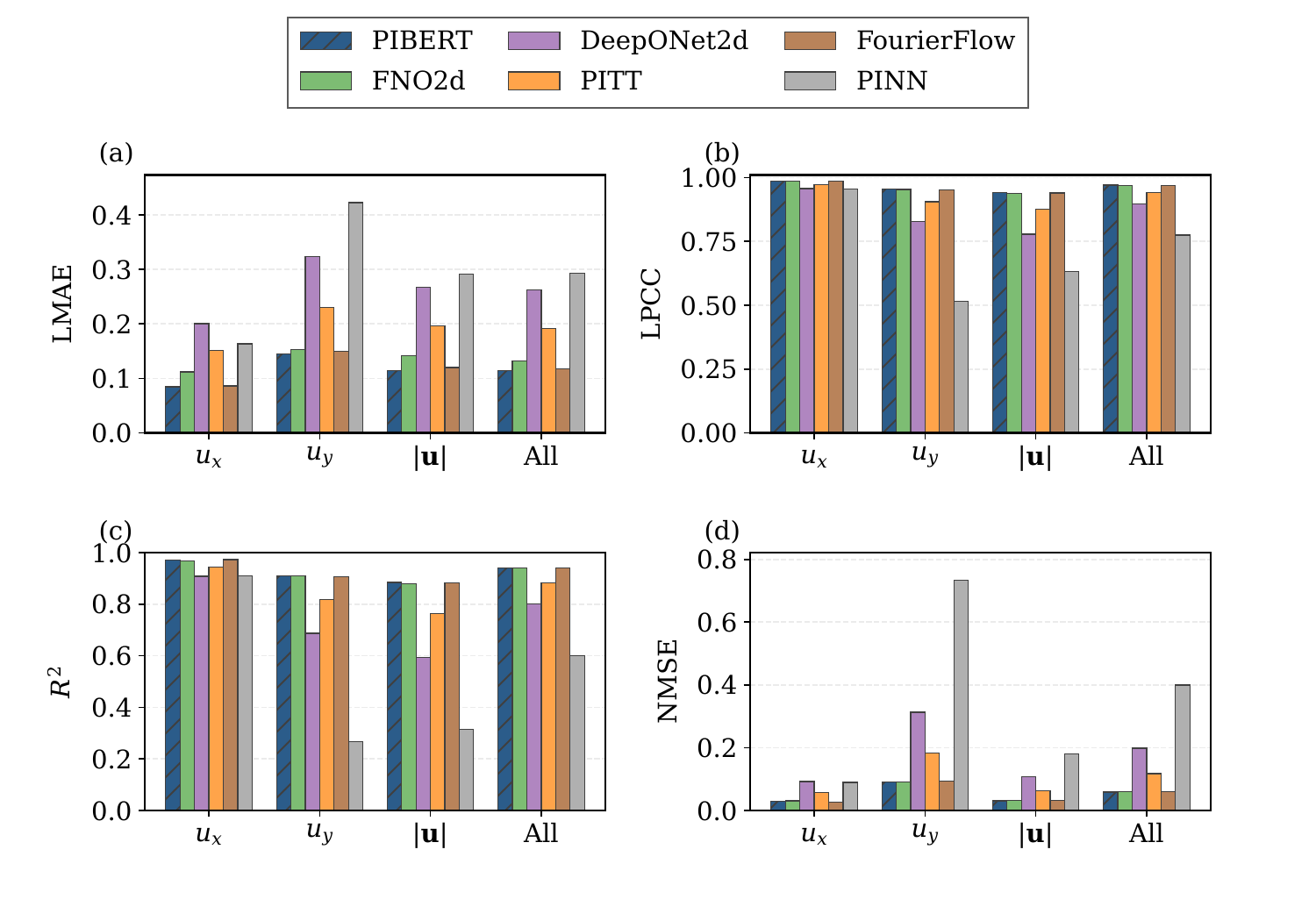}
    \caption{Aggregate performance on the RealPDEBench cylinder-real benchmark for PIBERT, FNO2d, DeepONet2d, PITT, FourierFlow, and PINN. Panel (a) reports LMAE, panel (b) reports LPCC, panel (c) reports $R^2$, and panel (d) reports NMSE for $u_x$, $u_y$, $|\mathbf{u}|$, and the aggregate ``All'' score over the cylinder-real test set. Lower is better for LMAE and NMSE, while higher is better for LPCC and $R^2$.}
    \label{fig:cyl-metrics}
\end{figure}

Figure~\ref{fig:cyl-metrics} gives the dataset-wide ranking on cylinder-real. PIBERT achieves the best aggregate accuracy with All NMSE $0.05875$ and All LPCC $0.97019$, and it is also the top model on aggregate LMAE and $R^2$. The advantage is especially clear in the cross-stream component, where PIBERT reaches $u_y$ NMSE $0.08984$ versus $0.73364$ for PINN. The best baselines remain competitive on selected sub-metrics. For example, FourierFlow is slightly better on the isolated $u_x$ NMSE, and FNO2d is close on $u_y$ but the full aggregate ranking favors PIBERT once both velocity components and $|\mathbf{u}|$ are considered jointly.

\begin{table}[h]
\centering
\caption{Dataset-wide aggregate accuracy summary for the RealPDEBench cylinder-real comparison set. Values are transcribed from the finalized  summary used for \Cref{fig:cyl-metrics}.}
\label{tab:cyl-aggregate-quant}
\begin{tabular}{lcccc}
\toprule
\tblhead{Model} & \tblhead{All NMSE} & \tblhead{All LMAE} & \tblhead{All LPCC} & \tblhead{All $R^2$} \\
\midrule
PIBERT & \textbf{0.05875} & \textbf{0.11448} & \textbf{0.97019} & \textbf{0.94123} \\
FourierFlow & 0.05931 & 0.11797 & 0.96989 & 0.94066 \\
FNO2d & 0.06003 & 0.13216 & 0.96959 & 0.93995 \\
PITT & 0.11729 & 0.19095 & 0.94118 & 0.88267 \\
DeepONet2d & 0.19846 & 0.26190 & 0.89658 & 0.80146 \\
PINN & 0.39938 & 0.29292 & 0.77492 & 0.60046 \\
\bottomrule
\end{tabular}
\end{table}

\begin{table}[h]
\centering
\caption{Cylinder-real component comparison between PIBERT and PINN. This table isolates the largest  component-wise contrast in the benchmark summary.}
\label{tab:cyl-component-quant}

\begin{tabular}{lcccc}
\toprule
\tblhead{Component} & \tblhead{PIBERT NMSE} & \tblhead{PINN NMSE} & \tblhead{PIBERT LPCC} & \tblhead{PINN LPCC} \\
\midrule
$u_x$ & \textbf{0.02993} & 0.08954 & \textbf{0.98490} & 0.95413 \\
$u_y$ & \textbf{0.08984} & 0.73364 & \textbf{0.95407} & 0.51642 \\
$|\mathbf{u}|$ & \textbf{0.03051} & 0.18036 & \textbf{0.94085} & 0.63154 \\
All & \textbf{0.05875} & 0.39938 & \textbf{0.97019} & 0.77492 \\
\bottomrule
\end{tabular}
\end{table}

\begin{table}[h]
\centering
\caption{Per-sample inference times for the cylinder-real comparison set. Values are reported in milliseconds per sample and are transcribed from the finalized cylinder-real accuracy--cost summary used for the supplementary performance-scatter figure.}
\label{tab:cyl-inference}
\begin{tabular}{lc}
\toprule
\tblhead{Model} & \tblhead{Inference time (ms/sample)} \\
\midrule
PIBERT & 5.85 \\
FourierFlow & 3.77 \\
FNO2d & 3.31 \\
PITT & 1.98 \\
DeepONet2d & 0.62 \\
PINN & 0.47 \\
\bottomrule
\end{tabular}
\end{table}

Tables~\ref{tab:cyl-aggregate-quant} and \ref{tab:cyl-component-quant} summarize the cylinder-real ranking and highlight PIBERT's stronger recovery of the cross-stream component and overall correlation structure relative to PINN.

\begin{figure}[htbp]
    \centering
    \includegraphics[width=\linewidth]{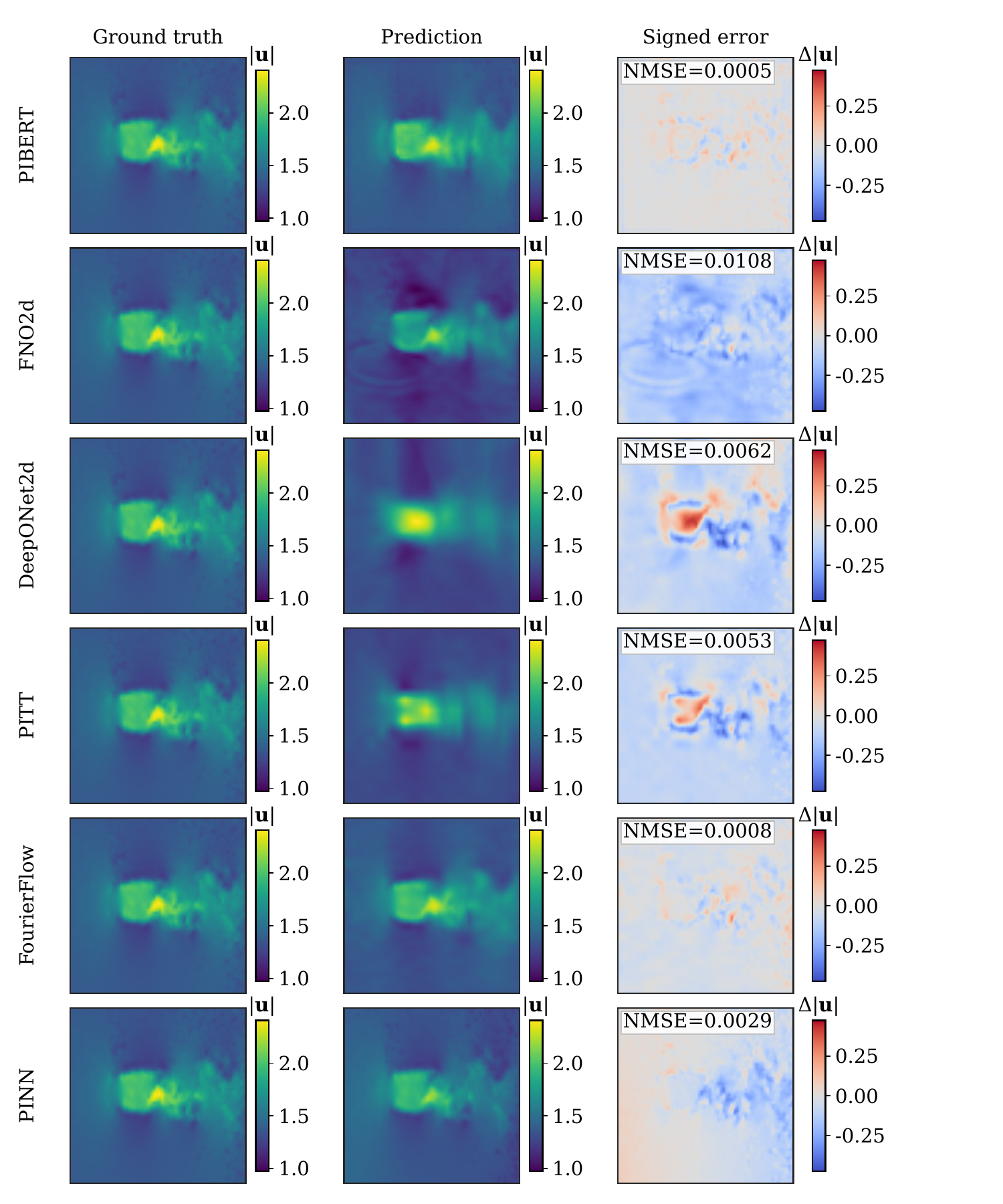}
    \caption{Comparison of ground truth, prediction, and signed error on a held-out RealPDEBench cylinder-real test sample across PIBERT, FNO2d, DeepONet2d, PITT, FourierFlow, and PINN for $|\mathbf{u}|$.}
    \label{fig:cyl-overview}
\end{figure}

\begin{figure}[htbp]
    \centering
    \includegraphics[width=\linewidth]{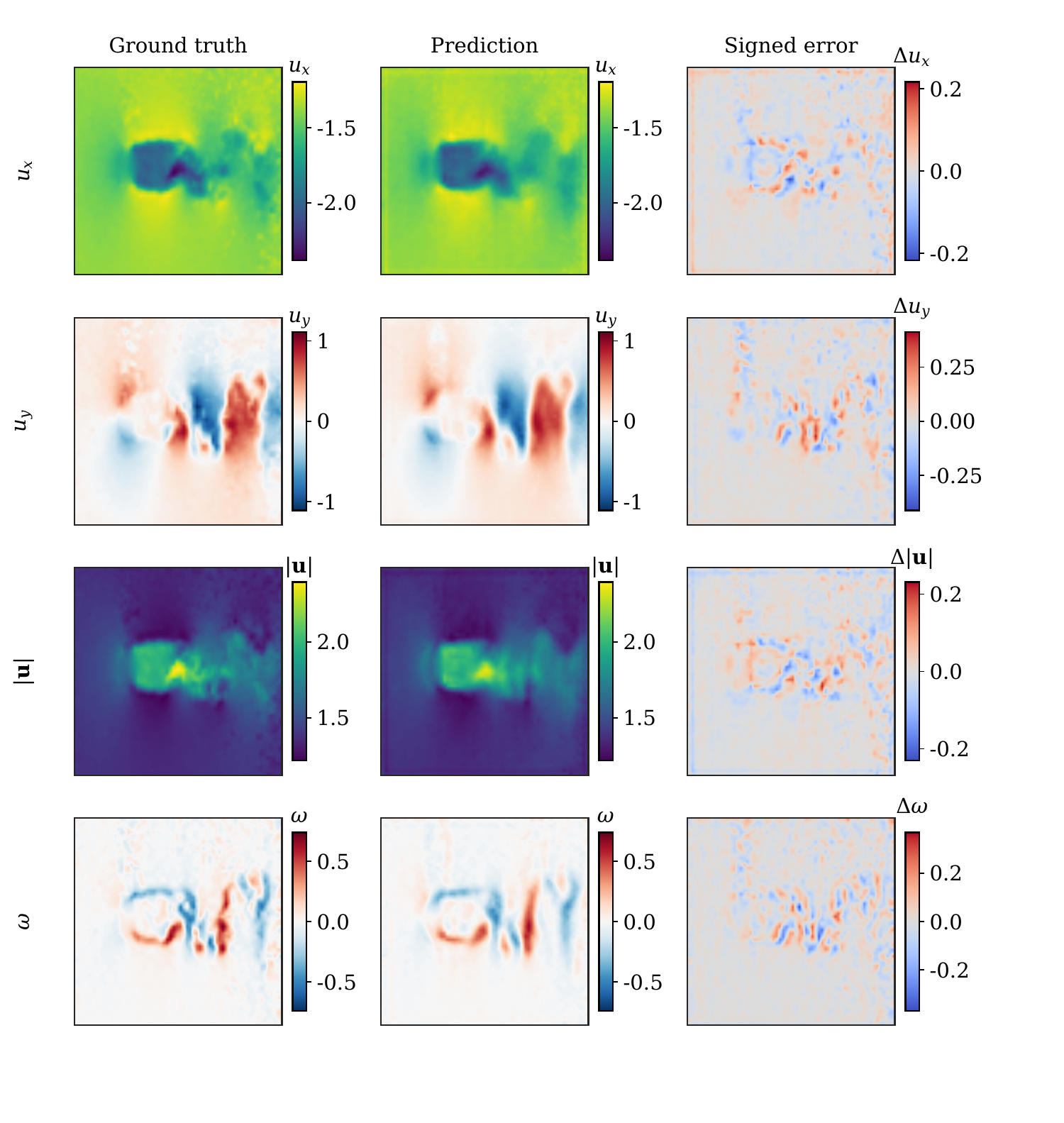}
    \caption{PIBERT prediction against ground truth, with signed error, for $u_x$, $u_y$, $|\mathbf{u}|$, and $\omega$ on the same held-out RealPDEBench cylinder-real sample.}
    \label{fig:cyl-pibert}
\end{figure}

Table~\ref{tab:cyl-inference} makes the inference-side tradeoff explicit for cylinder-real. PIBERT has the highest per-sample inference time in this comparison at $5.85$~ms/sample, versus $3.77$ for FourierFlow, $3.31$ for FNO2d, $1.98$ for PITT, $0.62$ for DeepONet2d, and $0.47$ for PINN. The cylinder-real ranking should therefore still be read as an accuracy comparison rather than an efficiency claim. We do not report a separate cylinder optimization table because the most complete optimization logs are available for FSI-real. The corresponding optimizer-side disclosure is moved to Appendix \Cref{tab:fsi-cost}.

The selected field panels in \Cref{fig:cyl-overview,fig:cyl-pibert} make the component-wise differences visible. PIBERT preserves the recirculation bubble, the downstream velocity deficit, and the cross-stream wake structure with smaller signed errors than the weaker baselines. Several operator-style baselines remain visually competitive on parts of the wake, whereas the weaker models more often smear or distort the wake-core structure. These panels are shown only to explain where the aggregate advantage appears. They do not replace the dataset-wide ranking in \Cref{fig:cyl-metrics}. The same interpretation is also supported by the supplementary Cylinder-real diagnostics. Figures S5.1--S5.3 show that PIBERT follows local phase and amplitude changes more consistently and preserves coherent wake structure across consecutive sampled steps. Figure S5.6 further shows that this pattern is not limited to isolated frames. 

\begin{figure}[htbp]
    \centering
    \includegraphics[width=\linewidth]{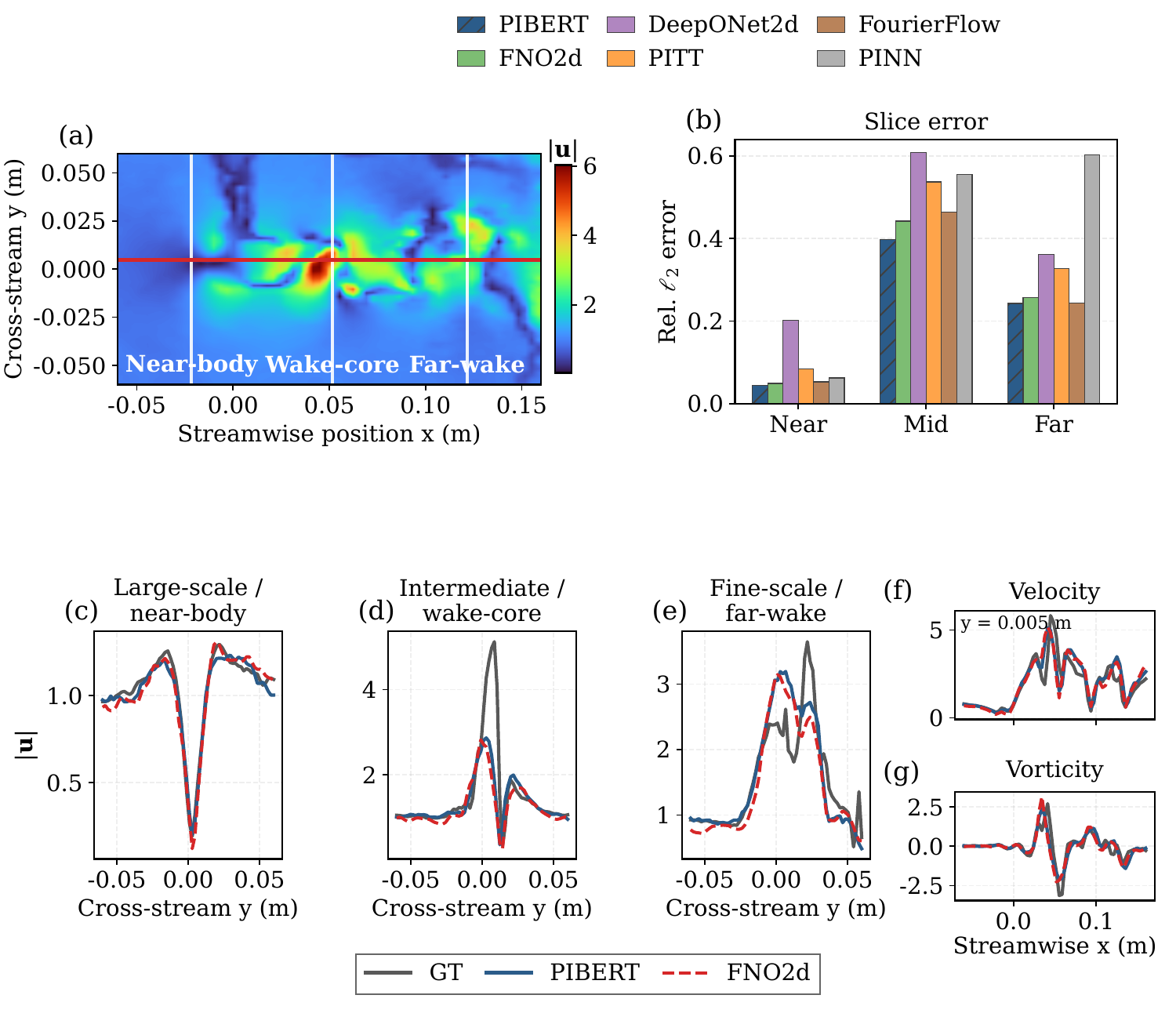}
     \caption{Scale-separated and wake-line diagnostics for the RealPDEBench cylinder-real $|\mathbf{u}|$ sample: (a) slicing layout, (b) slice relative $\ell_2$ error, (c) near-body profile, (d) wake-core profile, (e) far-wake profile, (f) wake-line velocity, and (g) wake-line vorticity.}
    \label{fig:cyl-multiscale}
\end{figure}

Figure~\ref{fig:cyl-multiscale} makes the multiscale claim more explicit than a contour plot alone can. On the selected $|\mathbf{u}|$ sample, PIBERT gives the lowest slice error in all three displayed regions and tracks the wake-line velocity and vorticity curves closely. The displayed baselines remain competitive on portions of the wake, but PIBERT is best on the three shown slice errors for this selected sample. Across the full cylinder-real test set of $1{,}990$ pairs, PIBERT wins $86.9\%$ of near-body slices, $40.7\%$ of wake-core slices, and $49.7\%$ of far-wake slices, with top-two rates of $98.4\%$, $68.3\%$, and $76.3\%$, respectively. The frequency-band story is intentionally more nuanced: PIBERT wins $40.8\%$ of low-band cases, $34.7\%$ of mid-band cases, and $23.8\%$ of high-band cases, so the manuscript does not claim universal spectral or high-frequency dominance.

The same caution applies to strict physics proxies. In the supplementary drag-proxy analysis, the ground-truth mean drag proxy is $0.05773$, while PIBERT reaches $0.06253$. This is strong and competitive, but FourierFlow ($0.06236$) and FNO2d ($0.06000$) are slightly closer on that single proxy. Accordingly, our cylinder-real claim is that PIBERT is the strongest aggregate-accuracy model with very strong near-body and wake-core fidelity, not that it dominates every strict spectral or wake-integral diagnostic.

\subsection{FSI-real benchmark}

\begin{figure}[htbp]
    \centering
    \includegraphics[width=\linewidth]{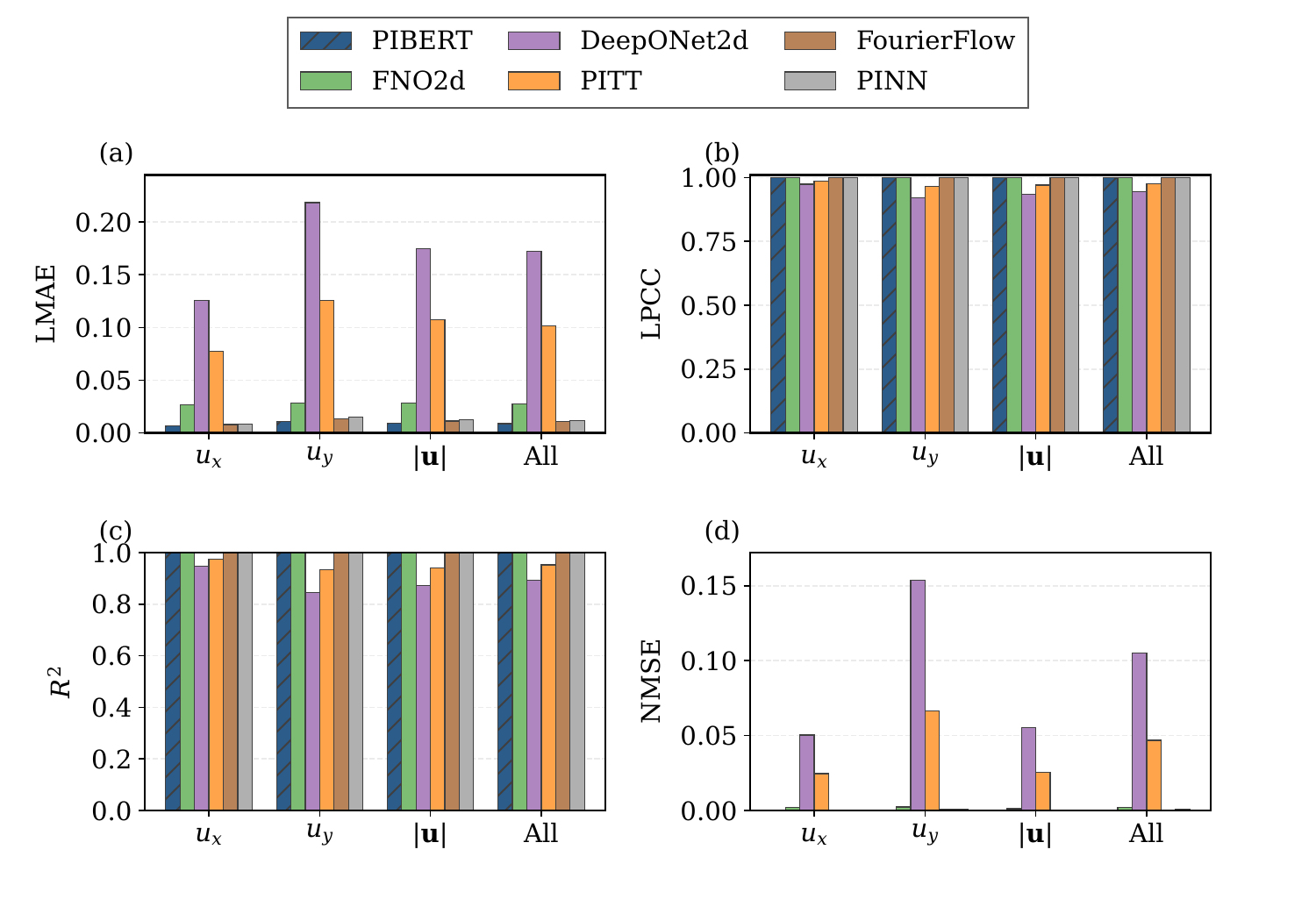}
    \caption{Aggregate performance on the RealPDEBench FSI-real benchmark for PIBERT, FNO2d, DeepONet2d, PITT, FourierFlow, and PINN. Panel (a) reports LMAE, panel (b) reports LPCC, panel (c) reports $R^2$, and panel (d) reports NMSE for $u_x$, $u_y$, $|\mathbf{u}|$, and the aggregate ``All'' score over the FSI-real test set. Lower is better for LMAE and NMSE, while higher is better for LPCC and $R^2$.}
    \label{fig:fsi-metrics}
\end{figure}

On FSI-real, PIBERT again leads the dataset-wide comparison. Figure~\ref{fig:fsi-metrics} shows All NMSE $0.00026954$ for PIBERT versus $0.00040231$ for the best baseline FourierFlow, with the remaining baselines trailing further behind. The advantage is consistent across the two  velocity components, and the same official real split and shared local $64\times 64$ learning representation are used for all six models. This result indicates best accuracy under the stated protocol, not blanket superiority on every auxiliary diagnostic.

\begin{table}[htbp]
\centering
\caption{Dataset-wide aggregate accuracy summary for the RealPDEBench FSI-real comparison set. Values are transcribed from the finalized  summary used for \Cref{fig:fsi-metrics}.}
\label{tab:fsi-aggregate-quant}
\begin{tabular}{lccccc}
\toprule
\tblhead{Model} & \tblhead{All NMSE} & \tblhead{All LMAE} & \tblhead{All LPCC} & \tblhead{All $R^2$} & \tblhead{All RelL2} \\
\midrule
PIBERT & \textbf{0.000270} & \textbf{0.008640} & \textbf{0.999864} & \textbf{0.999729} & \textbf{0.016418} \\
FourierFlow & 0.000402 & 0.010626 & 0.999797 & 0.999595 & 0.020058 \\
PINN & 0.000580 & 0.011629 & 0.999708 & 0.999416 & 0.024076 \\
FNO2d & 0.002248 & 0.027409 & 0.998873 & 0.997736 & 0.047416 \\
PITT & 0.046901 & 0.101337 & 0.976101 & 0.952772 & 0.216566 \\
DeepONet2d & 0.105176 & 0.172149 & 0.945567 & 0.894091 & 0.324309 \\
\bottomrule
\end{tabular}
\end{table}

\begin{table}[htbp]
\centering
\caption{FSI-real component comparison between PIBERT and the strongest  baseline FourierFlow. ``NMSE gain'' reports the relative reduction of PIBERT with respect to FourierFlow.}
\label{tab:fsi-component-quant}
\resizebox{\linewidth}{!}{%
\begin{tabular}{lccccccc}
\toprule
\tblhead{Component} & \tblhead{PIBERT NMSE} & \tblhead{FourierFlow NMSE} & \tblhead{NMSE gain} & \tblhead{PIBERT LPCC} & \tblhead{FourierFlow LPCC} & \tblhead{PIBERT $R^2$} & \tblhead{FourierFlow $R^2$} \\
\midrule
$u_x$ & \textbf{0.000143} & 0.000197 & 27.2\% & \textbf{0.999926} & 0.999898 & \textbf{0.999851} & 0.999796 \\
$u_y$ & \textbf{0.000381} & 0.000584 & 34.7\% & \textbf{0.999809} & 0.999708 & \textbf{0.999619} & 0.999416 \\
$|\mathbf{u}|$ & \textbf{0.000149} & 0.000226 & 34.2\% & \textbf{0.999829} & 0.999739 & \textbf{0.999657} & 0.999479 \\
All & \textbf{0.000270} & 0.000402 & 33.0\% & \textbf{0.999864} & 0.999797 & \textbf{0.999729} & 0.999595 \\
\bottomrule
\end{tabular}}
\end{table}

Tables~\ref{tab:fsi-aggregate-quant} and \ref{tab:fsi-component-quant} confirm that PIBERT is first on the benchmark-wide accuracy metrics and that its gain over FourierFlow remains visible on both velocity components, on speed magnitude, and on the aggregate score.

\begin{figure}[htbp]
    \centering
    \includegraphics[width=\linewidth]{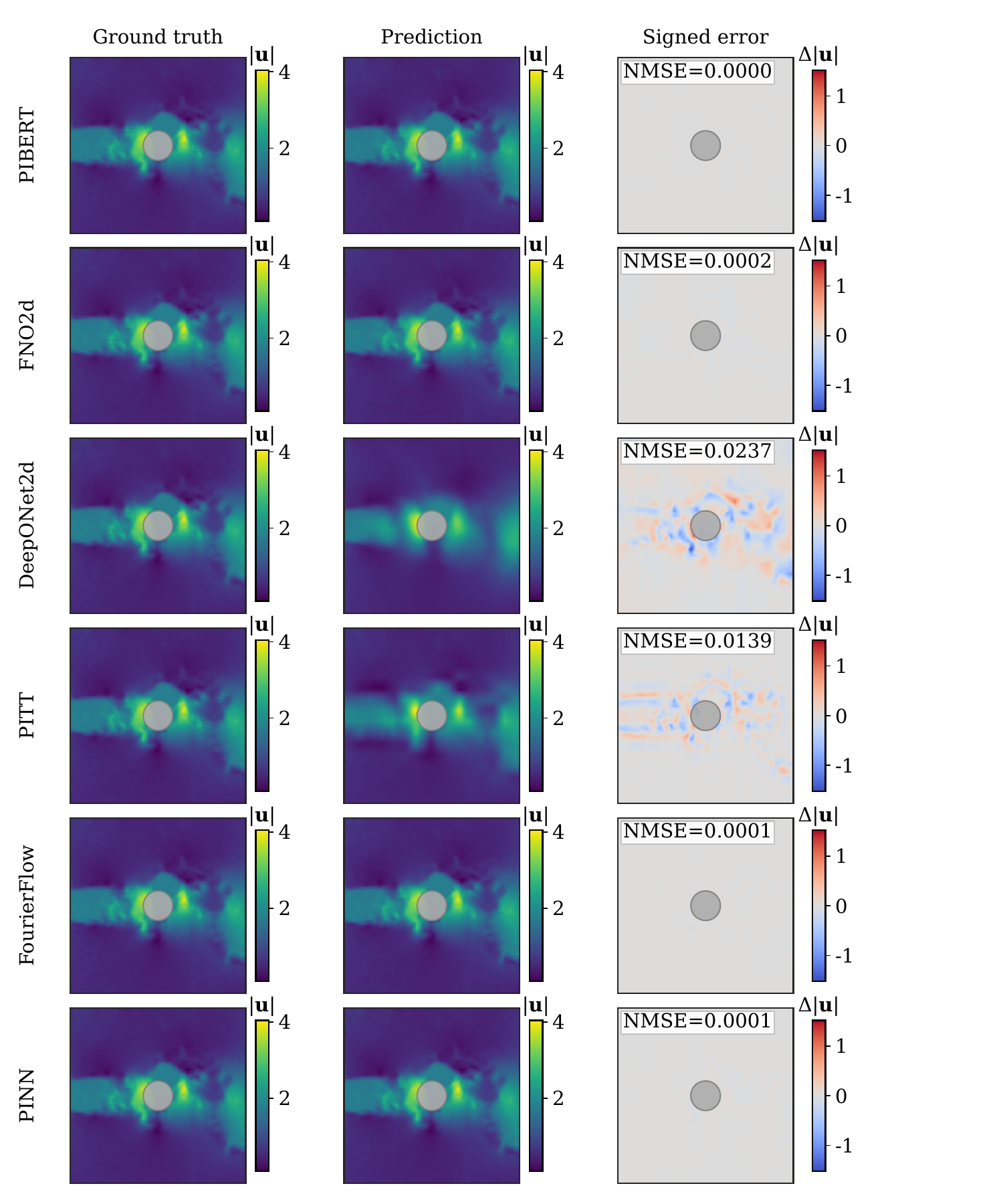}
    \caption{Comparison of ground truth, prediction, and signed error on a held-out RealPDEBench FSI-real test sample across PIBERT, FNO2d, DeepONet2d, PITT, FourierFlow, and PINN for $|\mathbf{u}|$.}
    \label{fig:fsi-overview}
\end{figure}

\begin{figure}[htbp]
    \centering
    \includegraphics[width=\linewidth]{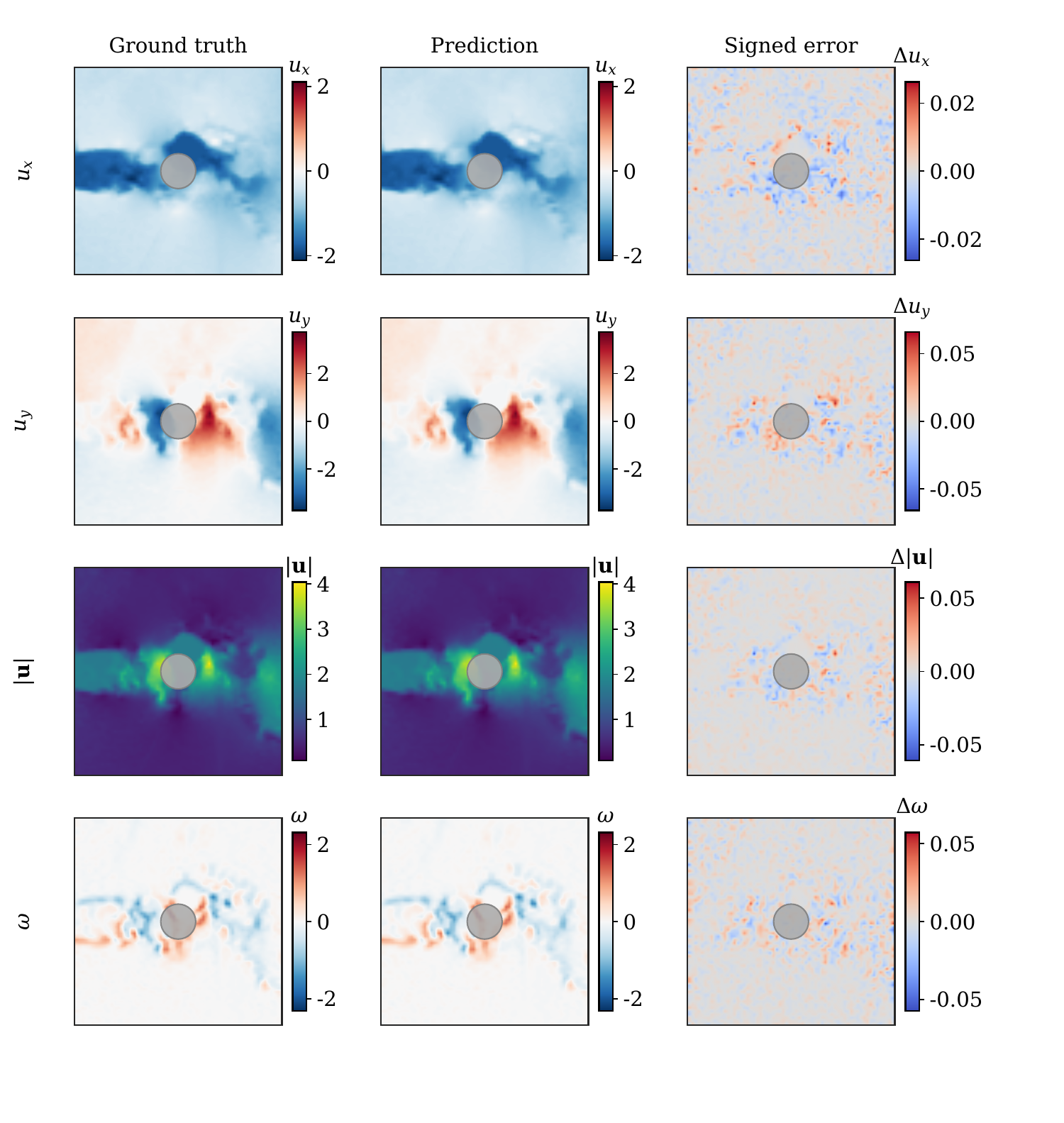}
    \caption{PIBERT prediction against ground truth, with signed error, for $u_x$, $u_y$, $|\mathbf{u}|$, and $\omega$ on the same held-out RealPDEBench FSI-real sample.}
    \label{fig:fsi-pibert}
\end{figure}

The FSI-real panels in \Cref{fig:fsi-overview,fig:fsi-pibert} are visually consistent with the aggregate metrics: PIBERT preserves the local wake structure around the tandem-cylinder configuration with the smallest signed-error footprint among the compared models. The stronger operator baselines remain visually competitive on parts of the field, while weaker baselines show larger structural distortions. As in the cylinder-real case, these panels explain \emph{where} the dataset-wide advantage appears and are not used as a stand-alone performance evidence.

\begin{figure}[htbp]
    \centering
    \includegraphics[width=\linewidth]{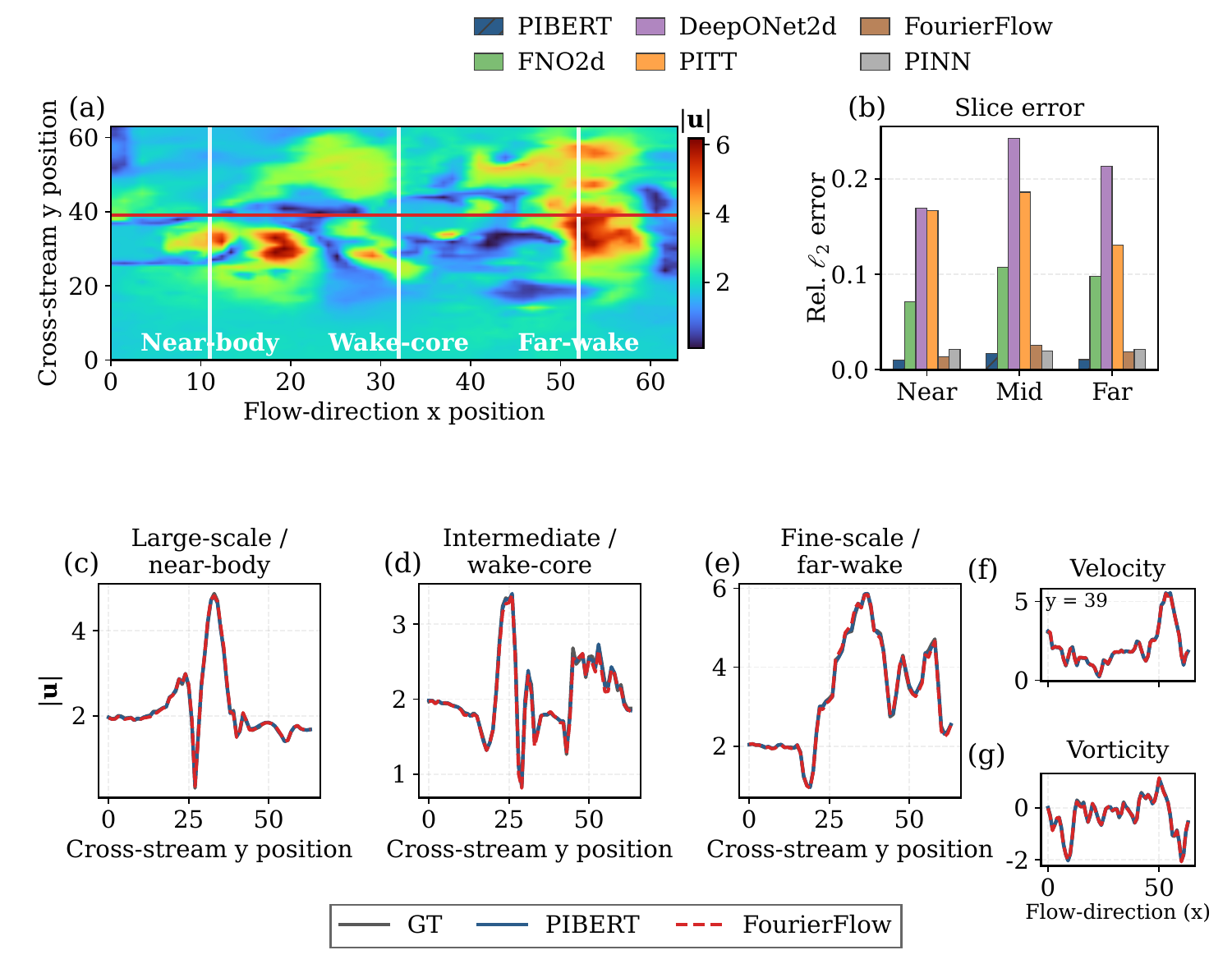}
     \caption{Scale-separated and wake-line diagnostics for the RealPDEBench FSI-real $|\mathbf{u}|$ sample: (a) slicing layout, (b) slice relative $\ell_2$ error, (c) near-body profile, (d) wake-core profile, (e) far-wake profile, (f) wake-line velocity, and (g) wake-line vorticity.}
    \label{fig:fsi-multiscale}
\end{figure}

Figure~\ref{fig:fsi-multiscale} localizes the FSI-real multiscale advantage on a  held-out sample. On the displayed slices, PIBERT gives the lowest relative $\ell_2$ error in all three shown regions: $0.0104$ near-body, $0.0167$ wake-core, and $0.0106$ far-wake, versus $0.0134$, $0.0255$, and $0.0188$ for the strongest displayed baseline FourierFlow. The same pattern appears on the wake-line traces, where PIBERT reaches velocity relative $\ell_2$ $0.0149$ versus $0.0219$ for FourierFlow and vorticity relative $\ell_2$ $0.0350$ versus $0.0475$. As with the cylinder-real multiscale figure, the panel is selected and is included to show \emph{where} the dataset-wide scale-separated advantage appears, not to replace the aggregated FSI-real evidence. We report the scale-separated benchmark table in our main text because it is the quantity most directly tied to the multiscale reconstruction claim while the optimizer-side is added in Appendix \Cref{tab:fsi-cost}.

\begin{table}[htbp]
\centering
\caption{Scale-separated and diagnostic summary for the FSI-real comparison set. Values are taken from the finalized FSI experiment logs and summary tables.}
\label{tab:fsi-multiscale}
\begin{tabular}{lccc}
\toprule
\textbf{Metric} & \textbf{PIBERT value} & \textbf{Best model} & \textbf{Best value} \\
\midrule
Scale-1 NMSE & 0.00026954 & PIBERT & 0.00026954 \\
Scale-2 NMSE & 0.00014469 & PIBERT & 0.00014469 \\
Scale-4 NMSE & 0.00006776 & PIBERT & 0.00006776 \\
Scale-8 NMSE & 0.00002994 & PIBERT & 0.00002994 \\
Vorticity MSE & 0.00018599 & PIBERT & 0.00018599 \\
Boundary MSE & 1.02852666 & PIBERT & 1.02852666 \\
Divergence MSE & 0.04030116 & DeepONet2d & 0.01261038 \\
Spectral L2 & 0.00017580 & FourierFlow & 0.00000198 \\
Spectral slope error & 0.00100403 & PINN & 0.00014111 \\
\bottomrule
\end{tabular}
\end{table}

Table~\ref{tab:fsi-multiscale} makes the dataset-wide multiscale claim explicit. PIBERT is best on all four coarse-to-fine aggregated NMSEs and also best on vorticity and boundary errors, while FourierFlow remains strongest on strict spectral L2 and PINN is best on the spectral-slope error. PIBERT also achieves the lowest first-level directional wavelet-detail NMSEs reported in the comparison set ($u$-LH1 $0.01003125$, $u$-HL1 $0.00265505$, $u$-HH1 $0.02804489$, $v$-LH1 $0.00297122$, $v$-HL1 $0.00435680$, $v$-HH1 $0.01495957$), which makes the scale-explicit advantage visible in FSI-real.

The dataset-wide slice diagnostics align with the selected panel in \Cref{fig:fsi-multiscale}. On the FSI multiscale evaluation set, PIBERT wins $87.9\%$ of near-body slices, $76.3\%$ of wake-core slices, and $91.4\%$ of far-wake slices, with corresponding top-two rates of $97.8\%$, $94.9\%$, and $98.6\%$. Its band win rates remain balanced across low, mid, and high bands ($48.5\%$, $55.4\%$, and $49.3\%$), with top-two rates of $85.2\%$, $86.0\%$, and $78.9\%$. This is the appropriate nuanced statement for FSI-real: PIBERT is the strongest model on the scale-separated reconstruction and local slice diagnostics, while FourierFlow remains the best strict spectral baseline. The supplementary FSI-real diagnostics show the same pattern in time. Figures S5.7--S5.10 show more consistent temporal tracking and more coherent spatial progression across consecutive sampled steps for PIBERT. Figures S5.11--S5.13 further show that the local-structure advantage remains visible across additional difficult held-out cases.

\subsection{Cross-benchmark multiscale interpretation}

\Cref{fig:cyl-multiscale,fig:fsi-multiscale,tab:fsi-multiscale} show where PIBERT performs better, especially in the near-body, wake-core, and multiscale flow regions. Across both benchmarks, the common pattern is more reliable reconstruction of the near-body region, wake-core organization, and downstream multiscale structure. On cylinder-real, this appears primarily as stronger recovery of separated wake geometry and cross-stream variation, while on FSI-real the same tendency becomes more systematic and extends across coarse-to-fine scale summaries, vorticity, and boundary-sensitive behavior. 

This distinction is physically important. In both benchmarks, we observe the hardest errors are not global amplitude mismatches. Rather, they arise in regions where advection, shear, recirculation, and geometry-induced interactions produce localized but dynamically important distortions. PIBERT performs best precisely on these diagnostics of local multiscale fidelity, whereas lighter baselines can still remain competitive on narrower spectral or integral-style summaries. Thus we conclude that in the cylinder-real benchmark  PIBERT is strongest on aggregate accuracy and wake-local reconstruction, even though FourierFlow and FNO2d remain highly competitive on selected spectral or drag-proxy-style checks. By contrast, the FSI-real evidence is stronger because the advantage persists not only in selected slices but also in the benchmark-level scale-separated summaries, indicating better preservation of coupled wake structure under a more spatially heterogeneous flow configuration. 

The cross-benchmark consistency also aligns with the methodological novelty of PIBERT. The hybrid Fourier-wavelet encoder represents both global oscillatory structure and localized distortions. The physics-aware tokenization and residual-biased self-attention then guide the model toward interactions that remain meaningful in physically active regions of the field. The MPP/ECP pretraining stage then reduces the dependence on learning these cross-scale relationships from limited supervised data alone. Overall, the results suggest that PIBERT's main contribution is that it more consistently preserves the physically consequential multiscale structure of real flow fields across two distinct benchmarks. The supplementary evidence extends this interpretation beyond the two main RealPDEBench benchmarks. The additional CFDBench, ICP Plasma, and EAGLE comparisons show lower error distributions and more localized error behavior for PIBERT relative to PINN in several non-identical physical settings. The Tube and Cavity embedding analyses further show that PIBERT maintains stronger fine-scale embedding--physics alignment, while the supplementary spectral diagnostics indicate better preservation of mid-to-high wavenumber content. These results are not used to replace the main RealPDEBench ranking, but they support the broader conclusion that the hybrid Fourier-wavelet encoder and physics-biased attention improve multiscale physical representation rather than only fitting two selected benchmark cases.

\section{Implications and Limitations}
\label{sec:implications}

The evidence now supports a narrower but stronger empirical claim. Across two RealPDEBench benchmarks, PIBERT is the best-accuracy model in the comparison once aggregate metrics, component-wise fields, multiscale slices, and local traces are considered together. On cylinder-real, the aggregate advantage is driven in large part by cross-stream structure recovery and near-body wake fidelity. On FSI-real, the gains persist under the official real split and remain visible across both field components and scale-separated diagnostics.

The evidence does \emph{not} support a claim that PIBERT is universally best on every diagnostic or that it is the cheapest model to run. Some narrow physics or frequency-domain summaries still favor lighter baselines, such as the cylinder-real drag proxy and the FSI strict spectral metrics. However, these isolated advantages do not change the main empirical pattern: PIBERT provides the strongest aggregate reconstruction and more consistent multiscale flow fidelity across the principal RealPDEBench metrics and diagnostics. The cost disclosure should therefore be interpreted as an accuracy-physics tradeoff rather than a weakness alone. PIBERT requires additional computation, but the added cost is tied to the model components that improve localized and scale-separated flow reconstruction.

The architectural template is also more general than the current PDE instantiation, but the residual bias itself is PDE-specific. For incompressible flow we use divergence and momentum-style residual diagnostics. For compressible flow, the same framework would need bias terms tied to mass, momentum, and energy balance, together with an equation-of-state closure and shock-aware diagnostics. For wave equations, the bias should reflect propagation structure, boundary reflections, and dispersion. For reaction--diffusion systems, the bias would need to encode reaction/diffusion imbalance together with conservation or positivity structure where appropriate. In that sense, the transformer scaffold is transferable, whereas the specific $R_{ij}$ construction must be redesigned for each PDE family. 

\begin{enumerate}
    \item \textbf{L1. Structured-grid dependence.} The experiments still operate on tensorized structured grids, even when the source benchmark is more complex than the local learning representation.
    \item \textbf{L2. Accuracy-cost tradeoff.} PIBERT requires more computation than lighter baselines, but this additional cost is associated with stronger aggregate reconstruction and better preservation of localized multiscale flow structures. Future work should reduce this cost while retaining the same physics-sensitive reconstruction behavior.
    \item \textbf{L3. PDE-specific residual bias.} The current attention bias is designed around incompressible-flow diagnostics and is therefore not a drop-in universal residual prior.
    \item \textbf{L4. Main benchmark scope and rollout horizon.} The primary claim-bearing evaluation is based on the cylinder-real and FSI-real RealPDEBench protocols, while supplementary CFDBench, ICP Plasma, EAGLE, Tube, and Cavity analyses provide broader supporting evidence. Longer-horizon rollout, direct structural-state prediction, and additional real-world engineering benchmarks remain open directions.
\end{enumerate}

\paragraph*{Matched future directions}
\begin{enumerate}
    \item \textbf{L1 $\rightarrow$ mesh-aware extensions.} Extend the current tensorized encoder to geometry-aware or unstructured discretizations so that the multiscale attention mechanism can operate on meshes and irregular domains without first collapsing them to a regular image grid.
    \item \textbf{L2 $\rightarrow$ sparse or hierarchical attention.} Reduce optimization cost through sparse, local, or hierarchical attention blocks, together with lighter multiscale fusion modules, so that the best-accuracy behavior is preserved at lower compute budgets.
    \item \textbf{L3 $\rightarrow$ modular PDE-family bias adapters.} Replace the single incompressible-flow residual module with modular bias terms tailored to compressible flow, waves, reaction--diffusion systems, and other PDE families.
    \item \textbf{L4 $\rightarrow$ coupled FSI states and longer rollouts.} Move beyond one-step flow reconstruction toward coupled flow--structure prediction, longer-horizon rollout, and uncertainty-aware trajectory modeling on real-world benchmarks.
\end{enumerate}

\section{Conclusion}
\label{sec:conclusion}
We introduced PIBERT as a transformer surrogate that integrates a hybrid Fourier-wavelet spectral encoder, physics-biased self-attention, and self-supervised MPP/ECP pretraining within a single architecture. These components remain the core methodological contribution of PIBERT.

The empirical framing is anchored to the RealPDEBench cylinder-real and FSI-real protocols and replaces contour-only arguments with a dataset-wide evidence hierarchy built from aggregate metrics, component-wise fields, multiscale slice diagnostics, local traces, and explicit cost disclosure. Under these protocols, PIBERT is the best-accuracy model in the comparison: on cylinder-real it reaches All NMSE $0.05875$ and All LPCC $0.97019$, and on FSI-real it reaches All NMSE $0.00026954$ versus $0.00040231$ for the best baseline FourierFlow. Additional supplementary analyses on CFDBench, ICP Plasma, EAGLE, Tube, and Cavity cases support the same interpretation: PIBERT's hybrid spectral encoding and physics-biased attention improve multiscale physical representation beyond global field fitting.

At the same time, we observe PIBERT is not the cheapest method in the comparison set, and strict spectral leadership is not universal even when the overall multiscale reconstruction is strongest. The central conclusion is therefore that PIBERT offers an physical accuracy advantage. It uses additional computation to recover localized and scale-separated flow structures more faithfully than lighter baselines under the studied protocols. Future work should preserve this multiscale reconstruction behavior while reducing cost through mesh-aware extensions, efficient attention mechanisms, PDE-family-specific residual adapters, and coupled FSI and longer-horizon prediction.


\appendix
\section{Extended Proofs}
\label{app:extended-proofs}

\begin{assumption}[2-D DFT/IFFT normalization]\label{asmp:dft-app}
For $x\in\C^{H\times W}$,
$\widehat{x}[k,\ell]=\sum_{m=0}^{H-1}\sum_{n=0}^{W-1} x[m,n]\,
e^{-2\pi i(\frac{mk}{H}+\frac{n\ell}{W})}$ and
$x[m,n]=\frac{1}{HW}\sum_{k,\ell}\widehat{x}[k,\ell]\,
e^{2\pi i(\frac{mk}{H}+\frac{n\ell}{W})}$. For real $x$, $\widehat{x}$ is Hermitian;
we use rFFT/irFFT storing the half-plane.
\end{assumption}

\begin{proposition}[Energy accounting with residual multiplier]\label{prop:residual-energy}
For any filter bank $\{K_s\}$,
$\sum_s\|x*K_s\|_2^2=\frac{1}{HW}\sum_{\omega}\big(\sum_s|\widehat{K}_s(\omega)|^2\big)\,|\widehat{x}(\omega)|^2$.
If $S(\omega)=\sum_s|\widehat{K}_s(\omega)|^2$, then
$\big|\sum_s\|x*K_s\|_2^2-\|x\|_2^2\big|\le \|S-1\|_\infty\,\|x\|_2^2$.
\end{proposition}

\begin{lemma}[Partition of unity]\label{lem:pou}
Let $g_0(\omega)=\cos\theta(\omega)$ and $g_1(\omega)=\sin\theta(\omega)$ with even $C^1$ $\theta:[-\pi,\pi]\to[0,\pi/2]$,
and define separable windows
$\widehat{K}_{\LL}=g_0(\omega_x)g_0(\omega_y)$,
$\widehat{K}_{\LH}=g_0(\omega_x)g_1(\omega_y)$,
$\widehat{K}_{\HL}=g_1(\omega_x)g_0(\omega_y)$,
$\widehat{K}_{\HHH}=g_1(\omega_x)g_1(\omega_y)$.
Then $\sum_{s\in\{\LL,\LH,\HL,\HHH\}}|\widehat{K}_s(\omega_x,\omega_y)|^2\equiv 1$.
\end{lemma}

\begin{lemma}[Discrete Green identity (periodic)]\label{lem:green-app}
For scalars $f,g$ on the 2-D torus, $\sum f\,(\Delta g)= -\sum \nabla f\cdot \nabla g$ with central differences.
\end{lemma}

\begin{proposition}[Quadratic boundary penalty enforces Dirichlet]\label{prop:dirichlet-penalty}
Let $J_\mu(u)=\|Au-f\|_2^2+\mu\|Bu-g\|_2^2$. Any minimizer $u_\mu$ converges, as $\mu\to\infty$, to the least-squares
solution of $Au=f$ subject to $Bu=g$.
\end{proposition}

\subsection{Fourier branch: $1$-Lipschitz and isometry}
\begin{proof}[Proof of \Cref{prop:fourier-energy}]
Adopt \Cref{asmp:dft-app}. Let $\widehat{x}(\omega)\in\C^{C}$ be the channel vector at frequency $\omega=(k,\ell)$.
The layer acts per-mode as
\[
\widehat{y}(\omega)=
\begin{cases}
W(\omega)\,\widehat{x}(\omega), & \omega\in\Omega_{\rm keep},\\
0, & \text{otherwise},
\end{cases}
\quad W(\omega)^\top W(\omega)=I .
\]
By Parseval,
$\|y\|_2^2=\tfrac{1}{HW}\sum_{\omega}\|\widehat{y}(\omega)\|_2^2
=\tfrac{1}{HW}\sum_{\omega\in\Omega_{\rm keep}}\|W(\omega)\widehat{x}(\omega)\|_2^2
=\tfrac{1}{HW}\sum_{\omega\in\Omega_{\rm keep}}\|\widehat{x}(\omega)\|_2^2
\le \tfrac{1}{HW}\sum_{\omega}\|\widehat{x}(\omega)\|_2^2=\|x\|_2^2.$
Hence the operator norm is $\le 1$ (nonexpansive). If $x$ is band-limited to
$\Omega_{\rm keep}$ then the inequality is an equality, i.e., an isometry.
\end{proof}

\subsection{Residual energy accounting and tight frame}
\begin{proof}[Proof of \Cref{prop:residual-energy}]
For any filter $K_s$, Parseval gives
$\|x * K_s\|_2^2=\tfrac{1}{HW}\sum_{\omega} |\widehat{K}_s(\omega)|^2 |\widehat{x}(\omega)|^2.$
Summing $s$ yields the stated identity and
\[
\Big|\sum_s\|x*K_s\|_2^2-\|x\|_2^2\Big|
=\frac{1}{HW}\sum_{\omega}\big|S(\omega)-1\big|\ |\widehat{x}(\omega)|^2
\le \|S-1\|_\infty\,\|x\|_2^2 .
\]
\end{proof}

\begin{proof}[Proof of \Cref{prop:tight-frame}]
By \Cref{lem:pou}, $\sum_s |\widehat{K}_s|^2\equiv 1$. Thus
$\sum_s\|x*K_s\|_2^2=\|x\|_2^2$ by Parseval.
Let $A$ be the analysis map $x\mapsto (K_s*x)_s$ and $S$ the synthesis
$S(y_s)=\sum_s K_s^\vee * y_s$. In the DFT basis, $A^\ast A$ has symbol
$\sum_s|\widehat{K}_s|^2\equiv 1$, hence $A$ is an isometry and $S A=I$; i.e.,
$x=\sum_s K_s^\vee*(K_s*x)$. Both analysis and synthesis are $1$-Lipschitz.
\end{proof}

\subsection{Hybrid fusion nonexpansiveness}
\begin{proof}[Proof of \Cref{prop:hybrid-nonexp}]
Let $\mathcal{F},\mathcal{W}$ satisfy $\|\mathcal{F}\|\le 1$, $\|\mathcal{W}\|\le 1$ and
$G_\alpha=\alpha\mathcal{F}+(1-\alpha)\mathcal{W}$ with $\alpha\in[0,1]$ (pointwise or spatially varying).
For any $x,y$,
\[
\|G_\alpha x - G_\alpha y\|
\le \alpha\|\mathcal{F}(x-y)\|+(1-\alpha)\|\mathcal{W}(x-y)\|
\le \alpha\|x-y\|+(1-\alpha)\|x-y\|=\|x-y\|.
\]
If both branches are isometries on the relevant subspace (e.g., band-limited input and $\mathcal{W}=I$),
then $\|G_\alpha x\|=\|x\|$.
\end{proof}

\subsection{Biased attention: ratio and Lipschitz bounds}
\begin{proof}[Proof of \Cref{lem:ratio}]
For a fixed row $i$,
$\alpha_{ij}=\exp(\tilde L_{ij})/\sum_m \exp(\tilde L_{im})$ with
$\tilde L_{ij}=L_{ij}-\lambda_{\rm att}R_{ij}$. Then
\[
\frac{\alpha_{ij_1}}{\alpha_{ij_2}}
=\exp\!\big(\tilde L_{ij_1}-\tilde L_{ij_2}\big)
=\exp\!\big((L_{ij_1}-L_{ij_2})-\lambda_{\rm att}(R_{ij_1}-R_{ij_2})\big),
\]
which is strictly decreasing in $\lambda_{\rm att}$ whenever $R_{ij_1}>R_{ij_2}$.
\end{proof}

\begin{proof}[Proof of \Cref{lem:lipschitz}]
Let $\alpha(\lambda)=\mathrm{softmax}(z-\lambda r)$ with row vectors $z,r$.
The Jacobian of softmax at $u$ is $J(u)=\mathrm{Diag}(\sigma(u))-\sigma(u)\sigma(u)^\top$.
By the mean value theorem,
$\alpha(\lambda)-\alpha(0)=\int_0^\lambda J(z-t r)(-r)\,dt$.
Using the operator norm $\|\cdot\|_{\infty\to 1}$,
$\|\alpha(\lambda)-\alpha(0)\|_1 \le \int_0^\lambda \|J(\cdot)\|_{\infty\to 1}\,dt\,\|r\|_\infty.$
One checks (e.g., by column sums) $\|J(\cdot)\|_{\infty\to 1}=\max_j 2\sigma_j(1-\sigma_j)\le \tfrac12$,
hence $\|\alpha(\lambda)-\alpha(0)\|_1\le \frac{\lambda}{2}\|r\|_\infty$.
Apply rowwise with $r=R_{i\cdot}$ and $\lambda=\lambda_{\rm att}$.
\end{proof}

\subsection{Translation equivariance and continuum limit}
\begin{proof}[Proof of \Cref{prop:equivariance}]
Let $\tau_s$ be the lattice shift by $s$. If $R_{ij}=\rho(p,r(i)-r(j))$, then
$L_{ij}$ and $R_{ij}$ shift compatibly: $L\circ\tau_s=\Pi_s L \Pi_s^\top$ and likewise for $R$,
with permutation matrix $\Pi_s$. Rowwise softmax commutes with the same permutation, so
$\alpha(\tau_s x)=\Pi_s \alpha(x) \Pi_s^\top$. Hence the mapping is translation-equivariant.
\end{proof}

\begin{proof}[Proof of \Cref{prop:continuum}]
Assume a periodic, compact $\Omega$ and bounded continuous $q(\cdot),k(\cdot),r(\cdot,\cdot)$.
On a grid of spacing $h$, the row $i$ softmax weights are
$w_h(x_i,y_j)=\frac{\exp(\langle q(x_i),k(y_j)\rangle-\lambda r(x_i,y_j))}{\sum_m \exp(\langle q(x_i),k(y_m)\rangle-\lambda r(x_i,y_m))}$.
Then $(T_h v)(x_i)=\sum_j w_h(x_i,y_j) v(y_j)\,h^d$ is a Riemann sum for
$(Tv)(x)=\int_\Omega w_\lambda(x,y) v(y)\,dy$ with the same normalized exponential kernel.
Uniform boundedness and continuity yield uniform convergence by dominated convergence; the normalization
enforces $\int w_\lambda(x,y)\,dy=1$.
\end{proof}

\subsection{Discrete Green identity and quadratic penalty}
\begin{proof}[Proof of \Cref{lem:green-app}]
For periodic central differences, $D_x^\top=-D_x$ and $D_y^\top=-D_y$.
With $\Delta=-(D_x^\top D_x+D_y^\top D_y)$,
\[
\sum f\,(\Delta g)= -\sum f\,D_x^\top D_x g - \sum f\,D_y^\top D_y g
= -\sum (D_x f)\,(D_x g) - \sum (D_y f)\,(D_y g).
\]
\end{proof}

\begin{proof}[Proof of \Cref{prop:dirichlet-penalty}]
The minimizer $u_\mu$ satisfies the normal equations
$(A^\top A+\mu B^\top B)u_\mu=A^\top f+\mu B^\top g$.
If $u^\star$ solves $Au=f$ with $Bu=g$ (in the least-squares sense), then
$u_\mu\to u^\star$ as $\mu\to\infty$ by standard quadratic-penalty arguments:
$B u_\mu\to g$ and $A u_\mu\to f$; any limit point solves the constrained problem.
\end{proof}

\subsection{Divergence-aware oracle inequality}
\begin{proof}[Proof of \Cref{prop:oracle-div}]
Define $\mathcal{R}_\lambda(g)=\E\|x-g\|_2^2+\lambda \E\|Dg\|_2^2$ and let $g_\lambda$ be a minimizer.
For any $g^\star\in\ker D$,
$\E\|x-g_\lambda\|^2+\lambda\E\|Dg_\lambda\|^2 \le \E\|x-g^\star\|^2$.
Rearrange to obtain
$\E\|x-g_\lambda\|^2 \le \E\|x-g^\star\|^2 - \lambda \E\|Dg_\lambda\|^2$.
Now, by $\dist(u,\ker D)\le c_H\|Du\|_2$, take $u=g_\lambda(\tilde{x})$ pointwise and average to get
$\E \dist(g_\lambda,\ker D)^2 \le c_H^2 \E\|D g_\lambda\|^2$.
Using $\dist(a,\mathcal{H})^2\le \|a-b\|^2$ with $b\in\mathcal{H}$ and choosing $b=g^\star$ yields
\[
\E\|x-g_\lambda\|^2 \le \E\|x-g^\star\|^2 + \frac{c_H^2}{\lambda}\,\E\|D g_\lambda\|^2,
\]
which matches the stated bound.
\end{proof}

\section{Fourier and Wavelet Derivatives}
\label{app:fourier-wavelet-derivs}

\paragraph{Fourier branch}
Let $y=\IFFT\big(\widehat{y}\big)$ with $\widehat{y}(\omega)=W(\omega)\,\widehat{x}(\omega)$
on the kept half-plane and $\widehat{y}(\omega)=0$ otherwise. For a real $x$ we store the rFFT
half-plane and enforce Hermitian symmetry.

Given an upstream spatial gradient $g=\partial\mathcal{L}/\partial y$ and its DFT
$\widehat{g}=\FFT(g)$:
\[
\frac{\partial \mathcal{L}}{\partial \widehat{x}(\omega)}
= W(\omega)^{\HH}\,\widehat{g}(\omega),\qquad
\frac{\partial \mathcal{L}}{\partial W(\omega)}
= \widehat{g}(\omega)\,\widehat{x}(\omega)^{\HH},\quad \omega\in\Omega_{\rm keep}.
\]
The spatial gradient is $\partial\mathcal{L}/\partial x=\IFFT\big(\partial\mathcal{L}/\partial \widehat{x}\big)$.
For rFFT, mirror the gradients to satisfy Hermitian symmetry and zero out non-kept modes.

\paragraph{Tight-frame branch}
Analysis coefficients $z_s=K_s * x$ and (optionally) synthesis $\tilde x=\sum_s K_s^\vee * z_s$.
For any branch loss $\mathcal{L}(z_s)$ with upstream gradients $h_s=\partial\mathcal{L}/\partial z_s$:
\[
\frac{\partial \mathcal{L}}{\partial x}=\sum_s K_s^\vee * h_s,\qquad
\frac{\partial \mathcal{L}}{\partial K_s}=h_s * x^\vee.
\]
In the Fourier domain:
\[
\frac{\partial \mathcal{L}}{\partial \widehat{x}}=\sum_s \overline{\widehat{K}_s}\,\Had\,\widehat{h}_s,\qquad
\frac{\partial \mathcal{L}}{\partial \widehat{K}_s}=\overline{\widehat{x}}\ \Had\ \widehat{h}_s.
\]

\paragraph{Wavelet windows}
With $\widehat{K}_{\LL}=g_0(\omega_x)g_0(\omega_y)$,
$\widehat{K}_{\LH}=g_0(\omega_x)g_1(\omega_y)$,
$\widehat{K}_{\HL}=g_1(\omega_x)g_0(\omega_y)$,
$\widehat{K}_{\HH}=g_1(\omega_x)g_1(\omega_y)$, $g_0=\cos\theta$, $g_1=\sin\theta$:
\[
\frac{\partial g_0}{\partial \theta}=-\sin\theta,\qquad
\frac{\partial g_1}{\partial \theta}=\cos\theta,\qquad
\frac{\partial \widehat{K}_{\LL}}{\partial \theta_x}=\big(-\sin\theta_x\big)g_0(\omega_y),\ \ldots
\]
Chain with $\partial\theta/\partial\phi$ if $\theta(\omega;\phi)$ is parametrized.
To preserve the partition-of-unity \(\sum_s|\widehat{K}_s|^2\equiv 1\), either (i) parametrize via a single
angle field $\theta(\omega)$ as above, or (ii) renormalize $K_s$ by $S(\omega)^{-1/2}$ with
$S=\sum_s|\widehat{K}_s|^2$ after each update.

\paragraph{Gating and fusion}
For $E=\alpha_F F+(1-\alpha_F)W$ with scalar or spatially varying $\alpha_F=\mathrm{softmax}(\gamma_F,\gamma_W)$,
the gradients are
\[
\frac{\partial\mathcal{L}}{\partial F}=\alpha_F\,\frac{\partial\mathcal{L}}{\partial E},\quad
\frac{\partial\mathcal{L}}{\partial W}=(1-\alpha_F)\,\frac{\partial\mathcal{L}}{\partial E},\quad
\frac{\partial\mathcal{L}}{\partial \gamma_F}=\alpha_F(1-\alpha_F)\,\Big\langle
\frac{\partial\mathcal{L}}{\partial E},\,F-W\Big\rangle.
\]

\paragraph{Notes on rFFT bookkeeping}
(i) Handle DC/NYQUIST lines once (no mirroring). (ii) When enforcing column-unitarity on $W(\omega)$,
apply it only on the kept set; set others to zero. (iii) Gradients on mirrored bins must be conjugate.

\section{Benchmark Provenance and Local Learning Pipeline}
\label{sec:pde-data}

This appendix documents only the two RealPDEBench benchmarks used in the manuscript \cite{hu2026realpdebench}.

\paragraph{Provenance note}
Benchmark-source numerical provenance is reported only as provided by the RealPDEBench release \cite{hu2026realpdebench}. The local learning pipeline then uses the released benchmark fields and explicitly states any tensorization, resizing, or channel-selection choices made by the local experiments. When the source benchmark does not document solver order, mesh order, or related numerical-scheme details, we do not invent them here.
The benchmark inventory is summarized in \Cref{tab:app-benchmark-summary}, the source-versus-local separation is detailed in \Cref{tab:app-benchmark-source-details}, model metadata are collected in \Cref{tab:app-run-settings}, and the two  benchmark schematics are collected in \Cref{fig:app-benchmark-schematics}. The individual panels \Cref{fig:app-cylinder-summary,fig:app-fsi-schematic} correspond to the cylinder-real and FSI-real cases, respectively.

\begin{table}[htbp]
\centering
\caption{Appendix benchmark summary for the two RealPDEBench benchmarks studied in this paper.}
\label{tab:app-benchmark-summary}
\resizebox{\linewidth}{!}{%
\begin{tabular}{p{2.2cm}p{3.3cm}p{3.2cm}p{3.6cm}p{3.4cm}}
\toprule
\textbf{Benchmark} & \textbf{Source benchmark scenario} & \textbf{Source channels / grid} & \textbf{Local learning task} & \textbf{Split used here} \\
\midrule
Cylinder-real & Real bluff-body wake benchmark in RealPDEBench & Real PIV $(u,v)$, with paired simulated $(u,v,p)$ available in the official package; benchmark description reports real $128\times 256$ and simulated $64\times 128$ & One-step prediction of real $(u,v)$ after local decoding, stride-20 sampling, and resizing to $64\times 64$ & Seed-42 split with $73/9/10$ trajectories and $14{,}527/1{,}791/1{,}990$ train/val/test pairs \\
FSI-real & Official tandem-cylinder vortex-induced-vibration benchmark with real observations & Released real flow fields at $128\times 128$ with benchmark metadata for the physical setting & One-step prediction of real-only $(u,v)$ after resizing to $64\times 64$ & Official real split shared across all six models \\
\bottomrule
\end{tabular}}
\end{table}

\begin{table}[htbp]
\centering
\caption{Appendix provenance details separating benchmark-source facts from the local learning pipeline.}
\label{tab:app-benchmark-source-details}
\resizebox{\linewidth}{!}{%
\begin{tabular}{p{2.2cm}p{4.1cm}p{4.0cm}p{4.5cm}}
\toprule
\textbf{Benchmark} & \textbf{Benchmark-source facts used in the manuscript} & \textbf{Local learning pipeline used in the experiments} & \textbf{What is intentionally not claimed} \\
\midrule
Cylinder-real & $92$ trajectories, $23{,}990$ frames, $20$~s at $400$~Hz, Reynolds numbers $1800$--$12000$, benchmark-level real $(u,v)$ and paired simulated $(u,v,p)$ modalities \cite{hu2026realpdebench} & Released real velocity tensors decode to $(3990,2,64,128)$ per trajectory in the local experiments, are sampled every 20 native steps, resized to $(2,64,64)$, and paired as $x_t \mapsto x_{t+1}$. The local parameter branch is zero-filled when a benchmark parameter JSON is not present in the released dataset. & We do not claim solver order, mesh order, or undisclosed CFD-generation details for the source benchmark because these are not specified in the official release. \\
FSI-real & Official real split of the tandem-cylinder vortex-induced-vibration benchmark, with released physical metadata and $128\times 128$ source fields \cite{hu2026realpdebench} & Local study uses the released real-only velocity channels $(u,v)$, resized to $64\times 64$, under the same six-model next-step prediction protocol. & We do not infer missing discretization or solver details beyond the source benchmark documentation, and we do not reinterpret the released physical metadata beyond what the benchmark states. \\
\bottomrule
\end{tabular}}
\end{table}
\begin{table}[htbp]
\centering
\caption{Appendix optimization-cost for the FSI-real comparison set.}
\label{tab:fsi-cost}
\resizebox{\linewidth}{!}{%
\begin{tabular}{lcccc}
\toprule
\tblhead{Model} & \tblhead{Parameters} & \tblhead{Checkpoint MB} & \tblhead{Last logged epoch (s)} & \tblhead{Wall time (min)} \\
\midrule
PIBERT & 8,714,248 & 99.51 & 49.2 & 13.9 \\
FourierFlow & 2,777,282 & 31.85 & 9.6 & 27.7 \\
FNO2d & 1,054,082 & 12.12 & 7.7 & 22.8 \\
PITT & 216,898 & 2.54 & 8.1 & 23.1 \\
DeepONet2d & 150,674 & 1.77 & 3.1 & 8.4 \\
PINN & 21,250 & 0.29 & 2.9 & 6.7 \\
\bottomrule
\end{tabular}}
\end{table}

\paragraph{Model metadata}
The experiment preserves parameter counts for all six comparison models, while the explicit hyperparameter block is reported for the finalized PIBERT configuration. We therefore report a hyperparameter table in \Cref{tab:app-run-settings}.

\begin{table}[htbp]
\centering
\caption{PIBERT hyperparameters for FSI and Cylinder. All reported configurations use the same model size of $1.560$M parameters.}
\label{tab:app-run-settings}
\begin{tabularx}{\linewidth}{p{2.5cm}Y}
\toprule
\tblhead{Group} & \tblhead{Hyperparameters / values} \\
\midrule
Architecture & channels $2\!\to\!2$; image size $64$; patch size $4$; embedding width $128$; depth $4$; heads $4$; MLP ratio $4.0$; Fourier modes $16$; parameter-token dim $16$; residual skip on; refinement hidden $32$ \\
Physics bias & $\lambda_{\mathrm{att}}=0.12$; $\alpha_{\mathrm{div}}=1.0$; $\alpha_{\mathrm{mom}}=1.0$; dropout $0.02$ \\
Loss weights & $\lambda_{\mathrm{div}}=1.0$; $\lambda_{\mathrm{lap}}=0.12$; $\lambda_{\mathrm{bnd}}=0.002$; $\lambda_{\mathrm{reg}}=5\times 10^{-5}$; $\lambda_{\mathrm{MPP}}=1.0$; $\lambda_{\mathrm{ECP}}=0.1$ \\
Optimization & batch size $256$; micro-batch size $128$; learning rate $5\times 10^{-5}$; minimum learning rate $1\times 10^{-7}$; warmup $3$; pretraining $50$ epochs; fine-tuning $120$ epochs; patience $45$; mask ratio $0.15$ \\
Data protocol & process cylinder-real; data type real; image size $64$; stride $20$; max frames per trajectory $200$; train/val ratios $0.8/0.1$; normalization on \\
\bottomrule
\end{tabularx}
\end{table}

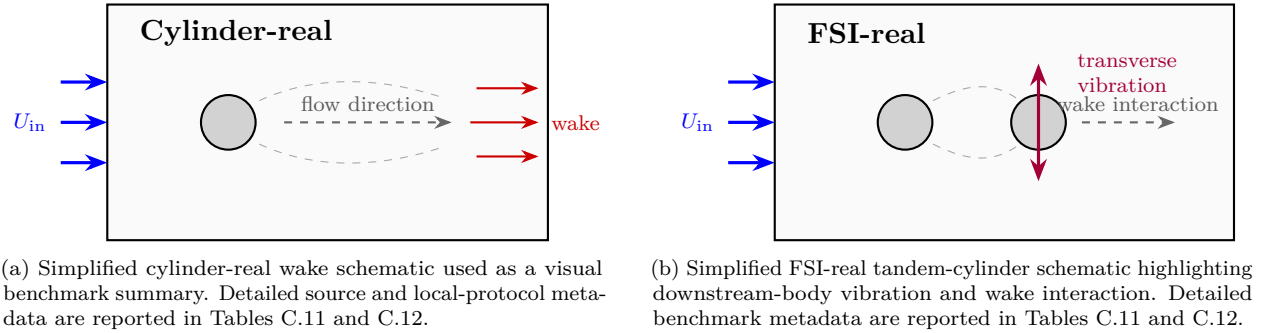
\begin{figure}[t]
\centering
\begin{subfigure}[b]{0.48\textwidth}
\centering
\begin{tikzpicture}[scale=0.82, >=Stealth, every node/.style={font=\scriptsize}]
  \fill[gray!4] (0,0) rectangle (7.1,3.8);
  \draw[thick] (0,0) rectangle (7.1,3.8);
  \node[anchor=west, font=\small\bfseries] at (0.35,3.35) {Cylinder-real};
  \filldraw[fill=gray!35, draw=black, thick] (1.95,1.9) circle (0.44);
  \draw[->, blue, very thick] (-0.75,1.25) -- (0,1.25);
  \draw[->, blue, very thick] (-0.75,1.9) -- (0,1.9);
  \draw[->, blue, very thick] (-0.75,2.55) -- (0,2.55);
  \node[left, blue] at (-0.8,1.9) {$U_{\mathrm{in}}$};
  \draw[->, thick, black!60, dashed] (2.85,1.9) -- (5.55,1.9);
  \node[above, black!65] at (4.2,1.9) {flow direction};
  \draw[->, red!80!black, thick] (5.95,1.35) -- (6.95,1.35);
  \draw[->, red!80!black, thick] (5.85,1.9) -- (6.95,1.9);
  \draw[->, red!80!black, thick] (5.95,2.45) -- (6.95,2.45);
  \node[right, red!80!black] at (6.98,1.9) {wake};
  \draw[black!35, dashed] (2.45,1.55) .. controls (3.3,1.15) and (4.35,1.15) .. (5.4,1.5);
  \draw[black!35, dashed] (2.45,2.25) .. controls (3.3,2.65) and (4.35,2.65) .. (5.4,2.3);
\end{tikzpicture}
\caption{Simplified cylinder-real wake schematic used as a visual benchmark summary. Detailed source and local-protocol metadata are reported in \Cref{tab:app-benchmark-summary,tab:app-benchmark-source-details}.}
\label{fig:app-cylinder-summary}
\end{subfigure}\hfill
\begin{subfigure}[b]{0.48\textwidth}
\centering
\begin{tikzpicture}[scale=0.82, >=Stealth, every node/.style={font=\scriptsize}]
  \fill[gray!4] (0,0) rectangle (7.4,3.8);
  \draw[thick] (0,0) rectangle (7.4,3.8);
  \node[anchor=west, font=\small\bfseries] at (0.35,3.35) {FSI-real};
  \filldraw[fill=gray!35, draw=black, thick] (2.1,1.9) circle (0.44);
  \filldraw[fill=gray!35, draw=black, thick] (4.25,1.9) circle (0.44);
  \draw[->, blue, very thick] (-0.75,1.25) -- (0,1.25);
  \draw[->, blue, very thick] (-0.75,1.9) -- (0,1.9);
  \draw[->, blue, very thick] (-0.75,2.55) -- (0,2.55);
  \node[left, blue] at (-0.8,1.9) {$U_{\mathrm{in}}$};
  \draw[<->, purple!85!black, very thick] (4.25,0.95) -- (4.25,2.85);
  \node[anchor=west, align=left, text=purple!85!black] at (4.7,2.7) {transverse\\vibration};
  \draw[->, thick, black!60, dashed] (4.95,1.9) -- (6.45,1.9);
  \node[above, black!65] at (5.85,1.9) {wake interaction};
  \draw[black!35, dashed] (2.55,1.55) .. controls (3.15,1.25) and (3.55,1.25) .. (3.95,1.55);
  \draw[black!35, dashed] (2.55,2.25) .. controls (3.15,2.55) and (3.55,2.55) .. (3.95,2.25);
\end{tikzpicture}
\caption{Simplified FSI-real tandem-cylinder schematic highlighting downstream-body vibration and wake interaction. Detailed benchmark metadata are reported in \Cref{tab:app-benchmark-summary,tab:app-benchmark-source-details}.}
\label{fig:app-fsi-schematic}
\end{subfigure}
\caption{Benchmark schematics used in the manuscript. Left: simplified cylinder-real wake configuration. Right: simplified tandem-cylinder vortex-induced-vibration FSI-real configuration. These panels are visual summaries only; benchmark-source metadata and local learning-pipeline details are reported in \Cref{tab:app-benchmark-summary,tab:app-benchmark-source-details}.}
\label{fig:app-benchmark-schematics}
\end{figure}

\subsection*{Local learning pipeline}

\paragraph{Cylinder-real local protocol}
The cylinder-real experiments use trajectory-level splits with seed $42$. Each decoded trajectory contains $3990$ native frames with two velocity channels and spatial size $64\times 128$. The local benchmark script samples every 20 native steps, keeps at most 200 frames per trajectory, resizes each frame to $64\times 64$, and forms consecutive one-step prediction pairs on the sampled sequence. This produces $14{,}527$ training pairs, $1{,}791$ validation pairs, and $1{,}990$ test pairs.

\paragraph{FSI-real local protocol}
The FSI-real experiments follow the official real split and use the released real-only velocity channels $(u,v)$ as the shared supervised target. For uniform comparison across PIBERT, FourierFlow, FNO2d, PITT, DeepONet2d, and PINN, the local protocol resizes the source fields to $64\times 64$ and evaluates all models on the same next-step tensorized prediction task.

\bibliographystyle{elsarticle-num}
\bibliography{references}

\end{document}